\documentclass[numberwithinsect,cleveref,a4paper,notab,hidelipics,thm-restate]{lipics-v2021}

\hideLIPIcs
\nolinenumbers

\usepackage{graphicx}
\usepackage{cite}
\usepackage{amsmath}
\usepackage{amssymb}
\usepackage{complexity}

\newcommand{\complex}{\mathscr{C}}
\newcommand{\surf}{\mathscr{S}}

\newcommand{\cccn}{\textsc{Color-Constrained Crossing Problem}}
\def\ca#1{\mathcal{#1}}
\def\OO{{\ca O}}
\def\cf#1{\mathscr{#1}}
\def\DD{{\cf D}}
\def\DxC#1{\textsc{#1-Crossing}}
\def\DC{\DxC{$\DD$}}
\def\DCP{\textsc{\DC{} Problem}}
\def\DCREM{\textsc{\DCP{} with Vertex Splits and Edge Removals}}
\def\DCROT{\textsc{\DCP{} with Flippable Rotation System}}

\definecolor{darkgreen}{rgb}{0,0.3,0}

\counterwithin{figure}{section}
\graphicspath{{figures/}{}}

\newcommand{\Prob}[3]{\vskip3pt%
\par\noindent\hspace*{-\fboxsep}\hspace*{-\fboxrule}%
\begin{nolinenumbers}%
\fbox{%
\begin{minipage}{\linewidth}
{\large\sc#1}\\
{\bf Input:\enspace}{#2}\\
{\bf Question:\enspace}{#3}
\end{minipage}%
}\hspace*{-\fboxsep}\hspace*{-\fboxrule}%
\end{nolinenumbers}%
\vskip3pt}

\title{A Unified FPT Framework for Crossing Number Problems}

\author{\'Eric Colin de Verdi\`ere}{LIGM, CNRS, Univ Gustave Eiffel, F-77454 Marne-la-Vall\'ee, France}{eric.colin-de-verdiere@univ-eiffel.fr}{}{Part of this work was funded by ANR project SUGAR (ANR-25-CE40-0416).}
\author{Petr Hlin\v{e}n\'y}{Masaryk University, Faculty of Informatics, Brno, Czech republic}{hlineny@fi.muni.cz}{https://orcid.org/0000-0003-2125-1514}{Part of this work was done while this author was invited professor at LIGM, Marne-la-Vallée, supported by Paris-Est Sup.}
\authorrunning{\'E.~Colin de Verdi\`ere and P.~Hlin\v{e}n\'y}


\acknowledgements{We thank Marcus Schaefer for helpful discussions on the many non-traditional crossing number variants, and the reviewers of the previous conference version for their many constructive suggestions.}

\Copyright{\'Eric Colin de Verdi\`ere and Petr Hlin\v{e}n\'y}

\relatedversion{An extended abstract appeared in \emph{Proceedings of the European Symposium on Algorithms} 2025.}

\DOIPrefix{}

\begin{document}
\maketitle

\begin{abstract}
  The basic (and traditional) \emph{crossing number problem} is to determine the minimum number of crossings in a topological drawing of an input graph in the plane. We develop a unified framework yielding fixed-parameter tractable (FPT) algorithms for many generalized crossing number problems.

  Our framework takes the following form.  We fix a surface~$\surf$ and a class $\DD$ of ``allowed'' topological drawings of graphs in~$\surf$ (e.g., some class of drawings with at most $t$ crossings).  We assume that testing membership in~$\DD$ can be done algorithmically, and that restricting a drawing in~$\DD$, extending it without adding any crossing, or transforming it with a self-homeomorphism of~$\surf$ yields a drawing that is also in~$\DD$.  Then deciding whether an input graph~$G$ has a drawing in~$\DD$, and computing one if it is the case, is fixed-parameter tractable in (essentially) the genus of~$\surf$ and the maximum number of crossings in a drawing in~$\DD$.  More generally, we may take as input an edge-colored graph and distinguish crossings by the colors of the involved edges; and we may allow a bounded number of edge removals and vertex splits on~$G$ before drawing it.  The proof is a reduction to the embeddability of a graph on a two-dimensional simplicial complex.

  This implies, in a unified way, linear or quadratic FPT algorithms for many topological crossing number variants established in the graph drawing community.  Some of these variants already had previously published FPT algorithms, mostly relying on Courcelle's metatheorem, but our algorithms have a better runtime.  Moreover, our framework extends to these crossing number variants in any fixed surface, and also allows us to fix the rotation system of the drawing of a graph in some variants.
\end{abstract}

\section{Introduction}

The (traditional) \emph{crossing number} problem, minimizing the number of pairwise edge crossings in topological drawings of an input graph in the plane,
is a long-standing and central task in graph drawing and visualization that comes in many established flavors; see the extensive dynamic survey by Schaefer~\cite{DBLP:journals/combinatorics/SchaeferDS21}.
In this paper, we give a framework that smoothly captures many crossing numbers variants and provides, in a unified way, efficient algorithms when parameterized by the solution value (i.e., the respective crossing number).

\subparagraph{Many flavors of crossing numbers.}

There is currently a surge of crossing number variants.  A motivation for such variations is that, in practical drawing applications, not every crossing or crossing pattern of edges may be ``equal'' to other ones.  
One may, e.g., want to avoid mutual crossings of important edges.  Or, to allow crossings only within specific parts of a graph, and not between unrelated parts.  Or, to exclude crossings of edges of some type.  Or, to allow only certain ``comprehensible'' crossing patterns in order to help visualize the graph.  
This is the main focus of the recent research direction called ``beyond planarity''~\cite{DBLP:journals/csur/DidimoLM19}.

Paraphrasing Schaefer~\cite[Chapter~2]{DBLP:journals/combinatorics/SchaeferDS21}, a crossing number variant minimizes, over all \emph{allowed drawings}~$D$ of an input graph~$G$ in some specified \emph{host surface}, an \emph{objective function} related to the crossings in~$D$.  The \emph{allowed drawings} may, e.g., restrict the crossings on one edge or in a local pattern of crossing edges, or prescribe or forbid certain local properties of the drawing (also depending on types or colors of edges).  The \emph{host surface} is usually the plane, but can be any fixed surface, orientable or not.  The \emph{objective function} is often the number of crossings between edges, but other possibilities include, e.g., the number of edges involved in crossings, or the number of pairs of edges that cross an odd number of times.

\subparagraph{Existing algorithms for the traditional version.}

One usually considers the decision version of the (traditional) crossing number problem: Given an input graph~$G$ of size~$n$, and~an~integer~$t$, does $G$ have crossing number at most~$t$?
This problem is already \NP-hard in the plane, as proved by Garey and Johnson in 1983~\cite{GareyJ83}, 
and even in very specific cases~\cite{Hlineny06,DBLP:journals/corr/abs-2406-18933,DBLP:journals/siamcomp/CabelloM13}.  Moreover, the problem is~\APX-hard~\cite{Cabello13}, and the best known polynomial-time approximation algorithm~\cite{DBLP:conf/stoc/ChuzhoyT22} gives an approximation factor that is subpolynomial in~$n$ only for bounded-degree graphs.

Thus, the main current focus is on algorithms that are fixed-parameter tractable (\FPT) in~$t$ --- the runtime has the form $f(t)\cdot\poly(n)$, where $f$ is a computable function of the \emph{parameter}~$t$.  For fixed~$t$, Grohe~\cite{DBLP:journals/jcss/Grohe04} has described an $\OO(n^2)$-time algorithm, and Kawarabayashi and Reed~\cite{DBLP:conf/stoc/KawarabayashiR07} have announced an $\OO(n)$-time algorithm.  Both rely on Courcelle's metatheorem~\cite{Courcelle90}, which automatically entails a huge dependency in the parameter~$t$.

Two recent approaches avoid resorting to Courcelle's theorem.  First, Colin de Verdi\`ere, Magnard, and Mohar~\cite{cmm-egtds-22} have studied the problem of embedding a graph in a \emph{two-dimensional simplicial complex} (2-complex for short), a topological space obtained from a surface by adding isolated edges and identifying vertices.  Among other motivations for this problem, they observe~\cite[Introduction]{cmm-egtds-22} that the crossing number problem reduces to the embeddability of the input graph in a certain 2-complex depending only on~$t$.  Then Colin de Verdi\`ere and Magnard~\cite{cm-faeg2-26} have shown that the embeddability of a graph of size~$n$ in a 2-complex of size~$C$ can be tested in time $2^{\poly(C)}\cdot n^2$, or in time $f(C)\cdot n$ for some function~$f$.  Hence, this results in $2^{\poly(t)}\cdot n^2$-time and $f(t)\cdot n$-time algorithms for the crossing number problem, which, moreover, extend to any fixed surface.

Second, even more recently, Lokshtanov, Panolan, Saurabh, Sharma, Xue, and Zehavi~\cite{Cr-SODA25} have given a $2^{\OO(t\log t)}\cdot n\,$-time algorithm for the traditional crossing number problem in the plane, by a reduction to bounded treewidth and dynamic programming, \mbox{all in time linear in~$n$}.

\subparagraph{Existing algorithms for other variants.}

For other flavors of crossing numbers, which are (typically) also \NP-hard, the literature on FPT algorithms is scarce; besides Pelsmajer, Schaefer,~and~\v{S}tefankovi\v{c} \cite{DBLP:conf/gd/PelsmajerSS07a} giving \FPT\ algorithms for the odd and pair crossing numbers, we are only aware of two very recent papers, which bring a general approach to multiple crossing number flavors and which we now detail.  (The various flavors of crossing numbers are defined in \Cref{sec:consequences}.)

First, M\"unch and Rutter \cite{patterncrossing24}, extending Grohe's approach \cite{DBLP:journals/jcss/Grohe04}, provide a framework for quadratic \FPT\ algorithms for the crossing number of several types of beyond-planar drawings of graphs in the plane, characterized by forbidden combinatorial crossing patterns. This includes the crossing number of $k$-planar, $k$-quasi-planar, min-$k$-planar, fan-crossing, and fan-crossing free drawings of a graph for any constant~$k$.

Second, Hamm, Klute, and Parada~\cite{DBLP:journals/corr/abs-2412-13092}, in a subsequent recent preprint, have announced a generalization of the previous framework~\cite{patterncrossing24}, also handling forbidden topological crossing patterns in drawing styles in the plane and, more importantly, bringing the possibility of handling predrawn parts of the input graph (as already known for the traditional crossing number~\cite{DBLP:conf/compgeom/HammH22}).  The dependency in the size~$n$ of the input graph is an unspecified polynomial.

The above results all rely on Courcelle's theorem.  Apart from that, we are only aware of sporadic results for isolated variants.  Kawarabayashi and Reed~\cite{DBLP:conf/stoc/KawarabayashiR07} and Jansen, Lokshtanov, and Saurabh~\cite{jls-nopa-14} both claim, without details, a linear \FPT\ algorithm to compute the skewness of a graph, and N{\"{o}}llenburg, Sorge, Terziadis, Villedieu, Wu, and Wulms~\cite{DBLP:conf/gd/NollenburgSTVWW22} give a non-uniform \FPT\ algorithm for the splitting number in surfaces.
All mentioned papers except \mbox{N{\"{o}}llenburg et\,al.~are restricted to the plane}.

\subparagraph{New contributions.}
Our primary contribution lies in developing a general and easily describable framework that unifies a large part of the many existing variants of the crossing number problem.
Our framework takes the following form.  We fix a surface~$\surf$ and a class~$\DD$ of (possibly edge-colored) topological drawings of graphs on~$\surf$, corresponding to the class of ``allowed'' drawings.
We assume that testing membership in~$\DD$ can be done algorithmically and that restricting a drawing in~$\DD$, extending it without adding any crossing, or transforming it with a self-homeomorphism of~$\surf$ yields a drawing that is also in~$\DD$.  For example, $\DD$ could be the class of drawings of a given drawing style with at most~$t$ crossings, for some fixed integer~$t$.  

We prove that deciding whether an input (possibly edge-colored) graph~$G$ has a drawing in~$\DD$, and computing one if it is the case, reduces to (several) tests of embeddability of graphs on 2-complexes, discussed above~\cite{cm-faeg2-26}.  As a consequence, it can be done in quadratic time in the size~$n$ of~$G$, and exponential in a polynomial in the other parameters, namely the genus of~$\surf$, the number of colors, and the maximum number~$r$ of crossings in a drawing in~$\DD$; alternatively, the decision problem can be solved in time linear in~$n$ if the other parameters are fixed.  All of this remains true if one allows a bounded number of (possibly color-restricted) edge removals and vertex splits to~$G$ before drawing it; these numbers are also parameters.

We deduce linear \FPT\ algorithms for many established crossing number variants; we now survey some of them.
First, we can restrict ourselves to many established drawing styles, such as $k$-planar, $k$-quasi-planar, min-$k$-planar, $k$-gap, fan-crossing, weakly fan-planar, strongly fan-planar, $k$-cover planar, $k$-matching planar, and fan-crossing free drawings. In this setting, $k$ is formally an additional fixed parameter, but it does not affect the runtime since we can always assume $k\leq t$.  
For the cases of the $k$-gap, $k$-cover planar, and $k$-matching planar crossing numbers, no \FPT\ algorithms were known. 
Second, we may assign colors to the edges of the input graph and count the crossings differently depending on the colors of the edges involved, leading to the first \FPT\ algorithms for the joint crossing number on surfaces and its generalizations.
In particular, color-based restrictions on the crossings can be used to fix the rotation system of a graph (the cyclic ordering of the edges around each vertex), a property that is not yet addressed by any of the existing \FPT\ algorithms.
Third, assuming suitable additional properties, we can handle problems that do not count the crossings, but rather the number of edges, or pairs of edges, involved in crossings, such as the edge, pair, and odd crossing numbers (no \FPT\ algorithm was known for the edge crossing number).
And fourth, since we allow for prior edge removals and vertex splits, our framework encompasses, e.g., the skewness, for which \FPT\ algorithms were only sketched~\cite{DBLP:conf/stoc/KawarabayashiR07,jls-nopa-14} in the plane, and the splitting number, which was only known to admit a nonuniform \FPT\ algorithm~\cite{DBLP:conf/gd/NollenburgSTVWW22}.

Last but not least, all these four aspects that define flavors of crossing numbers can be combined arbitrarily, and our main results are valid not only in the plane, but for arbitrary surfaces.
As a side note, we allow non-orientable surfaces in our framework, homeomorphic to a sphere with a number of ``crosscaps''.  Placing $k$~crosscaps in the plane amounts to choosing $k$ specific points in the plane that can be traversed by an arbitrary number of pairwise crossing edges of the planar drawing ``for free'', without accounting for the induced crossings at these points, which seems relevant from a graph drawing perspective.
Also, we allow surfaces with boundary, which are needed for some subtle variations of the fan-planar crossing number.

For most flavors of crossing number, the arguments are direct, though they really differ according to the flavor; they essentially boil down to checking that the class~$\DD$ of drawings satisfies the desired conditions, and to giving an algorithm to test membership in~$\DD$.

\subparagraph{Comparison with the state of the art.}

Our main contribution is a convenient, versatile, and unified framework capturing a very general class of crossing number variants while giving \FPT\ algorithms.  In particular, as listed above, we obtain new \FPT\ results for several~established variants and their generalizations (to surfaces, or by combining the features of multiple variants).  
It is conceivable that some of the variants newly covered by our framework could be proved \FPT\ using other techniques, perhaps Courcelle's powerful metatheorem, as already used in \cite{DBLP:journals/jcss/Grohe04,DBLP:conf/compgeom/HammH22,patterncrossing24,DBLP:journals/corr/abs-2412-13092}.  However, such developments are extremely delicate even in isolated variants and lead to hard-to-describe algorithms.  Indeed, first, checking a given drawing and counting its crossings must be formalized in the MSO$_2$ logic of graphs, which usually needs heavy tricks tailored to the specific case; second, there remains the specific and often highly nontrivial treewidth reduction step to be done.  So far, such approaches have been successful only in the case of the plane.

Moreover, all approaches using Courcelle's theorem inherently come with a high computati\-onal cost.  Indeed, they result in algorithms with a huge multi-level exponential dependency in the parameter~$r$; for instance, it is an exponential tower of height at least four in the case of \cite{DBLP:journals/jcss/Grohe04} and three in the case of \cite{DBLP:conf/stoc/KawarabayashiR07}, see the discussion in~\cite{Cr-SODA25}.  Also, while Courcelle's theorem itself runs in linear time in the size~$n$ of the input graph, the necessary treewidth reduction step requires quadratic time, if not more, in the known approaches for the non-traditional crossing numbers.  In contrast, we obtain quadratic \FPT\ algorithms with a singly exponential dependency in the parameters, and alternative linear \FPT\ algorithms with unspecified dependency in the parameters.  
Furthermore, any improved algorithm for the embeddability problem~\cite{cm-faeg2-26} (or to the irrelevant vertex algorithm by Golovach, Kolliopoulos, Stamoulis, and Thilikos~\cite{gkst-fivlt-25}, used in the linear-time algorithm) would automatically speed up our algorithms.

The frameworks used in the two very recent works mentioned above~\cite{patterncrossing24,DBLP:journals/corr/abs-2412-13092}, which both use treewidth reduction and Courcelle's theorem, can handle some of the drawing styles that we encompass, but different counting functions, vertex splits, and nonplanar surfaces are out of their reach.
However, it should not be difficult to extend their framework to handle edge-colored graphs, thanks to the MSO$_2$ logic being able to handle arbitrary sets of edges.  On the other hand, the recent breakthrough of Lokshtanov et al.~\cite{Cr-SODA25} for the traditional crossing number problem may fuel hope for more $2^{\OO(t\log t)}\cdot n$-time algorithms, but tweaking all ingredients of that highly technical paper seems hard, even for isolated crossing number variants, and moreover this approach is inapplicable to surfaces other than the plane, due to the use of several inherently planar techniques (3-connectivity arguments and dependency~on~\cite{jls-nopa-14}).

\subparagraph{Organization of the paper.}

Starting with the preliminaries (\Cref{sec:prelim}), we describe our general framework and state our main result (\Cref{sec:metasection}), which is then proved (\Cref{sec:proofmeta}).  After that, we instantiate this framework to many crossing number variants, and list the resulting FPT algorithms (\Cref{sec:consequences}).  However, some crossing number variants that fix the rotation system of the input graph require more delicate arguments; we present another general tool, also deduced from our main result, and apply it to several crossing number problems in which the rotation system is fixed (\Cref{sec:metasection-rot}).  Finally, we show how to compute drawings corresponding to positive instances (\Cref{sec:compute}) and conclude (\Cref{sec:concl}).

\section{Preliminaries}\label{sec:prelim}

\subparagraph{Graphs and surfaces.}

In this paper, graphs are finite and undirected, but not necessarily simple unless specifically noted. The \emph{size} of a graph $G$ is the number of vertices plus the number of edges of~$G$.  A \emph{vertex split} of a graph~$G$ at vertex~$v$ creates a new vertex~$v'$ and replaces some of the edges incident to~$v$ by making them incident to~$v'$ instead of~$v$.

We follow~\cite{a-bt-83,mt-gs-01} for surface topology.  A \emph{surface}~$\surf$ is a topological space obtained from finitely many disjoint solid, two-dimensional triangles by identifying some of their edges in pairs.  Surfaces are not necessarily connected.  The \emph{boundary} of~$\surf$ is the closure of the union of the unidentified edges.  The \emph{(Euler) genus} of a connected surface is twice its number of ``handles'', if it is orientable, and is its number of ``crosscaps'', otherwise.
Up to homeomorphism, each \emph{connected} surface~$\surf$ is specified by whether it is orientable or not, by its (Euler) genus, and by its number of boundary components.  In this paper, by `genus' we always mean the Euler genus.
The \emph{topological size} of~$\surf$ equals $s:=d+g+b$, where $d$ is the number of connected components of~$\surf$, $g$ is the sum of the (Euler) genera of its connected components, and $b$ is the number of boundary components.  Algorithmically, a surface (up to a homeomorphism) can be specified by the genus, number of boundary components, and orientability of each connected component.  The plane is not a surface according to our definition, but in this paper it can always be substituted with a disk, which is (homeomorphic to) a surface.  A \emph{self-homeomorphism} of~$\surf$ is just a homeomorphism from~$\surf$ to~$\surf$. 

\subparagraph{(Topological) drawings of graphs.}

In this paper, any drawing of any graph in any topological space maps distinct vertices to distinct points, and has finitely many intersection points, each avoiding the images of the vertices; however, an intersection point may involve two or more pieces of edges.  In more formal terms, and introducing more terminology, the previous sentence means the following.
A \emph{curve}~$c$ in a topological space~$\cal X$ is a continuous map from the closed interval $[0,1]$ into~$\cal X$.  The \emph{relative interior} of~$c$ is the image, under~$c$, of the open interval~$(0,1)$.  In a {\em drawing}~$D$ of a graph $G$ in a topological space~$\cal X$, vertices are represented by points of~$\cal X$ and edges by curves in~$\cal X$ such that the \emph{endpoints} $c(0)$ and~$c(1)$ of a curve~$c$ are the images of the end vertices of the corresponding edge.  Moreover, we assume that distinct vertices are mapped to distinct points and that the relative interiors of the edges avoid the images of the vertices.  Given a point $x\in\cal X$, its \emph{multiplicity} in~$D$ is the number of pairs $(c,t)$ such that $c(t)=x$, where $c$ is a curve representing an edge of~$G$ in~$D$, and $t\in(0,1)$.  An \emph{intersection point} (shortly an \emph{intersection}) of~$D$ is a point with multiplicity at least two.  We only consider drawings~$D$ with finitely many intersection points.

When considering a drawing~$D$ of a graph~$G$ on a \emph{surface}~$\surf$, we \emph{always} implicitly assume that $G$ has no isolated vertex and that the image of~$D$ avoids the boundary of~$\surf$; these conditions are actually benign in a graph drawing context.  In~$D$, each edge~$e$ is subdivided into (finitely many) \emph{pieces} by the intersection points along~$e$.  A \emph{crossing} is an intersection of multiplicity two that cannot be removed by a local perturbation of the curves.  The other intersections of multiplicity two are \emph{tangencies}.  The drawing~$D$ is \emph{normal}~\cite{schaefer2017crossing} if every intersection point is a crossing (in particular, it has multiplicity two).  Moreover, $D$ is \emph{simple} if it is normal, no edge self-crosses, no two adjacent edges cross, and no two edges cross more than once.%
\footnote{Note that many authors do not consider a drawing $D$ \emph{simple} if $D$ contains two parallel edges (even uncrossed), as two parallel edges necessarily intersect at their two endpoints. Likewise, uncrossed loops are usually not allowed in simple drawings. We do not make these additional restrictions here, for reasons that will be apparent in the proof of \Cref{prop:kplanar-etc}.}

The (traditional) \emph{crossing number} of a graph $G$ in a surface $\surf$ is the least number of crossings over all normal drawings of $G$ in~$\surf$.  Many variations exist~\cite{DBLP:journals/combinatorics/SchaeferDS21}, see \Cref{sec:consequences}.

\subparagraph{Storing drawings on surfaces.}\label{sec:storing}

We need to represent drawings of graphs on a surface~$\surf$ combinatorially, up to self-homeomorphisms of~$\surf$.  Cellular embeddings can be represented by combinatorial maps~\cite{e-dgteg-03}, but we need to allow intersection points, and moreover faces need not be disks.  We can achieve this, ultimately relying on combinatorial maps.
Here are the details.

We represent drawings on surfaces as subsets of the vertex-edge graph of some triangulation of~$\surf$.  (An alternative possibility would be through \emph{extended combinatorial maps}~\cite[Section~2.2]{cm-tgis-14}, but this appears to be unnecessarily complicated for our purposes.)

Specifically, a \emph{triangulation} of a surface~$\surf$ is a graph~$T$ embedded on~$\surf$ such that each face of~$T$ (each connected component of the complement of the image of~$T$) is homeomorphic to a disk bounded by three edges, and all three incident vertices and edges are pairwise distinct.  In particular, a triangulation is a 2-complex.  Let $T'$ be the graph made of the vertices and edges of~$T$ that lie in the interior of~$\surf$.  Let $W=(w_1,\ldots,w_k)$ be a set of $k$~walks in~$T'$, such that each edge of~$T'$ is used at most once by the union of these walks.  This set~$W$ naturally gives rise to a drawing~$D$ of a graph~$G$ on~$\surf$: The graph~$G$ has $k$~edges, each drawn as $w_1,\ldots,w_k$ on~$\surf$; the vertices of~$G$ are the at most~$2k$ endpoints of the walks in~$W$.  We say that the pair $(T,W)$ is a \emph{representation} of~$D$.  

A representation of a colored drawing is a representation $(T,W)$ of the corresponding uncolored drawing, together with the data of a color for each walk in~$W$.  Its \emph{size} is the size of a data structure to represent it, which is $\OO(tc)$, where $t$ is the number of triangles and $c$ is the number of colors (assuming we store the colors in unary).  Conversely (as proved in Lemma~\ref{lem:represent-drawings} below), every drawing~$D$ on~$\surf$ can be represented by such a pair~$(T,W)$ in which $T$ is made of $\OO(s+u)$ triangles, where $s$ is the topological size of~$\surf$, and the intersection points of~$D$ subdivide the edges of~$D$ into $u$ pieces in total.

\subparagraph{2-complexes.}

\begin{figure}
    \centering
    \def\svgwidth{\linewidth}
\begingroup%
  \makeatletter%
  \providecommand\color[2][]{%
    \errmessage{(Inkscape) Color is used for the text in Inkscape, but the package 'color.sty' is not loaded}%
    \renewcommand\color[2][]{}%
  }%
  \providecommand\transparent[1]{%
    \errmessage{(Inkscape) Transparency is used (non-zero) for the text in Inkscape, but the package 'transparent.sty' is not loaded}%
    \renewcommand\transparent[1]{}%
  }%
  \providecommand\rotatebox[2]{#2}%
  \newcommand*\fsize{\dimexpr\f@size pt\relax}%
  \newcommand*\lineheight[1]{\fontsize{\fsize}{#1\fsize}\selectfont}%
  \ifx\svgwidth\undefined%
    \setlength{\unitlength}{294.24751025bp}%
    \ifx\svgscale\undefined%
      \relax%
    \else%
      \setlength{\unitlength}{\unitlength * \real{\svgscale}}%
    \fi%
  \else%
    \setlength{\unitlength}{\svgwidth}%
  \fi%
  \global\let\svgwidth\undefined%
  \global\let\svgscale\undefined%
  \makeatother%
  \begin{picture}(1,0.38338202)%
    \lineheight{1}%
    \setlength\tabcolsep{0pt}%
    \put(0,0){\includegraphics[width=\unitlength,page=1]{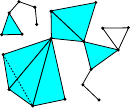}}%
    \put(0.1641941,0.00896089){\color[rgb]{0,0,0}\makebox(0,0)[lt]{\lineheight{1.49999988}\smash{\begin{tabular}[t]{l}(a)\end{tabular}}}}%
    \put(0.73923162,0.00896089){\color[rgb]{0,0,0}\makebox(0,0)[lt]{\lineheight{1.49999988}\smash{\begin{tabular}[t]{l}(b)\end{tabular}}}}%
  \end{picture}%
\endgroup%

    \caption{(a): An example of a two-dimensional simplicial complex, or 2-complex for short. (b): An embedding of some graph with four vertices and three edges on the same 2-complex.}
    \label{fig:2complex}
\end{figure}

In this paper, a \emph{$2$-complex}~$\complex$ (or two-dimensional simplicial complex) is a topological space obtained from a simple graph (without loops or multiple edges) by attaching solid, two-dimensional triangles to some of its cycles of length three; see \Cref{fig:2complex}.%
\footnote{Our definition of 2-complex slightly departs from the standard one; it is the same as a geometric simplicial complex of dimension at most two, realized in some ambient space of dimension large enough.} %
The \emph{simplices} of~$\complex$ are its vertices, edges, and triangles.
We can easily represent 2-complexes algorithmically, by storing the vertices, edges, and triangles, and the incidences between them.  In general, many such representations correspond to the same topological space, and the choice of the representation only impacts the complexity analysis of our algorithms.  The class of 2-complexes is quite general; it contains all graphs, all surfaces, all $k$-books (although in this paper, we only need 2-complexes in which every edge is incident to at most two triangles), and any space obtained from a surface by identifying finitely many finite subsets of points and by adding finitely many edges between any two points.  

An \emph{isolated edge} of a 2-complex~$\complex$ is an edge incident to no triangle.  The \emph{surface part} of~$\complex$ is the union of all its triangles, together with their incident vertices and edges.  A \emph{singular point} of~$\complex$ is a point that has no open neighborhood homeomorphic to an open disk, a closed half-disk, or an open segment.

Our above definition of drawings in topological spaces applies to drawings in 2-complexes.  In particular, vertices of~$G$ may lie anywhere on~$\complex$, and edges of~$G$ as curves may traverse several vertices, edges, and triangles of~$\complex$.  An \emph{embedding} of~$G$ into~$\complex$ is a drawing without any intersection point.

\subparagraph{Embeddability of graphs on 2-complexes. }

The \emph{embeddability problem} takes as input a graph~$G$ and a 2-complex~$\complex$, and the task is to decide whether $G$ has an embedding in~$\complex$.  This problem is fixed-parameter tractable in the size of the input 2-complex:

\begin{theorem}[{Colin de Verdière and Magnard~\cite[Theorems 1.1 and~11.1]{cm-faeg2-26}}]
\label{thm:embedding}
  One can solve the embeddability problem (of graphs in 2-complexes) in $2^{\poly(C)}\cdot n^2$ time or in $f(C)\cdot n$ time, where $C$ is the number of simplices of the input $2$-complex, $n$ is the size of the input graph, and $f$ is some function.
\end{theorem}

\section{The framework: description and main result}\label{sec:metasection}

Let $c$ be a positive integer and $\surf$ a surface.  We consider \emph{colored} graphs, in which each \emph{edge} is labeled with an integer in~$\{1,\ldots,c\}$.  A \emph{colored drawing} is a drawing~$D$ of a colored graph~$G$ in~$\surf$ in which each edge of~$D$ inherits the color of the corresponding edge of~$G$.

Let $\DD$ be a class of colored drawings on~$\surf$.  (Thus $\DD$ implicitly depends on~$c$ and~$\surf$.)  We say that $\DD$ is \emph{stable} if, for any $D\in\DD$, any drawing obtained from~$D$ by any of the following operations is also in~$\DD$:
  \begin{enumerate}[(a)]
  \item the removal of a vertex or an edge;
  \item the addition of a vertex whose image is disjoint from the image of~$D$ (and from the boundary of~$\surf$);
  \item the addition of an edge, of an arbitrary color, connecting two vertices in~$D$, and otherwise disjoint from the image of~$D$ (and from the boundary of~$\surf$);
  \item a self-homeomorphism of~$\surf$.
  \end{enumerate}
  For example, for any integer~$t$, the class of normal drawings with at most $t$ crossings is stable.

\smallskip
  For a stable class~$\DD$ of colored drawings, we introduce the following problem:
\Prob{\DCP{}}%
{a positive integer $c$; a colored graph~$G$ with colors in~$\{1,\ldots,c\}$}%
{Is there a colored drawing~$D\in\DD$ of the graph~$G$?}

Before proceeding further, we remark that the traditional crossing number problem is an instance of the \DCP{}:
\begin{observation}\label{obs:normalcr}
For every integer $t$, the problem whether a graph $G$ has crossing number at most~$t$ on a surface $\surf$ is a \DxC{$\DD_t$} instance for $c=1$ (a single color) and the stable class~$\DD_t$ of all normal drawings with at most $t$ crossings on $\surf$.
\end{observation}

Our framework is actually much richer than formulated above; for any stable class~$\DD$ of drawings, we introduce the following generalization of the \DCP{}:
\Prob{\DCREM{}}%
{a positive integer $c$; non-negative integers $p$, $q_1,q_2,\ldots,q_c$; a colored graph~$G$ with colors in~$\{1,\ldots,c\}$}%
{Is there a colored drawing~$D'\in\DD$ of a colored graph~$G'$ obtained from~$G$ by performing at most $p$ successive vertex splits and by removing, for $i=1,\ldots,c$, at most $q_i$ edges of color~$i$?}
\smallskip

To successfully tackle these problems in the parameterized setting, we need a suitable parameter bounding the ``richness of crossings'' in a stable class~$\DD$.  The \emph{intersecting size} of a drawing is the sum of the multiplicities of its intersection points.  The \emph{intersecting size} of~$\DD$ is the maximum intersecting size of a drawing in~$\DD$.  For example, if as above $\DD_t$ is the class of normal drawings with at most $t$ crossings, then the intersecting size of~$\DD_t$ is at most $2t$, because by definition each intersection point in a normal drawing is a crossing, which has multiplicity two.

Our main result is that, assuming an algorithm to decide membership in~$\DD$ and an upper bound on the intersecting size of~$\DD$, the \DCREM{} is fixed-parameter tractable (and, as a special case, also the \DCP{}):
\begin{theorem}\label{thm:main}
  Let $c$ be a positive integer.  Let $\surf$ be a surface of topological size~$s$.   Let $\DD$ be a stable class of drawings with colors in~$\{1,\ldots,c\}$ on~$\surf$, specified by an integer $r$ that is an upper bound on the intersecting size of~$\DD$, and by an algorithm that, in time $\OO(\delta(j))$, decides membership in~$\DD$ of a representation of a colored drawing of size at most~$j$.

  Then, the \DCREM{} can be solved in time $2^{\poly(c+s+r+p+q)}\cdot\delta(\OO((s+r)c))\cdot n^2$, or in $f(c+s+r+p+q)\cdot n$ time for some function~$f$, where $n$ is the size of~$G$ and $q:=\sum_iq_i$.
  In particular, the \DCP{} can be solved in time $2^{\poly(c+s+r)}\cdot\delta(\OO((s+r)c))\cdot n^2$, or in time $f(c+s+r)\cdot n$.
\end{theorem}

A few notes on \Cref{thm:main} are in order:
\begin{enumerate}
    \item More precisely, the proof of \Cref{thm:main} shows that an instance of the \DCREM{} can be solved by:
    \begin{itemize}
        \item $(c+s+r)^{\OO(s+r)}$ calls to an oracle that decides membership in~$\DD$, for drawing representations of size $\OO((s+r)c)$, and
        \item 
        $(c+s+r)^{\OO(s+r)}\cdot (s+r+p+q)^{\OO(p+q)}$ tests of embeddability of a graph of size~$\OO(c+r+p+q)n)$ into a 2-complex of size $\OO(s+(c+r+p+q)(r+q))$.
    \end{itemize}
    \item The upper bound~$r$ on the intersecting size of~$\DD$ shows up in the runtime, but is also used by the algorithm and must be provided on the input.  If $\DD$ contains only normal drawings, then we can take $r$ to be twice the maximum number of crossings in a drawing in~$\DD$, as explained above.
    \item In all our applications, the function $\delta(\cdot)$ will be at most exponential, and so the corresponding factor $\delta(\OO((s+r)c))$ (which is independent of the input graph size~$n$) in the runtime can be ignored.
    \item Our algorithms are uniform, in the sense that there is a single quadratic (respectively, linear) algorithm valid for all choices of $c$, $\surf$, $\DD$, $r$, $p$, $q_1,\ldots,q_c$, and~$G$.  The only caveat is that the algorithm that decides membership in~$\DD$ may obviously depend on~$c$, $\surf$, and~$\DD$.
    \item For positive instances of the \DCREM{}, we can compute an actual represen\-tation of the corresponding drawing without overhead compared to the quadratic-time algorithm; see \Cref{thm:explicit}.
\end{enumerate}

\section{Proof of \Cref{thm:main}}\label{sec:proofmeta}

Recall that $\surf$ is a surface, $c$ is a positive integer, $\DD$ is a stable class of of colored drawings, and $r$ is an upper bound on the intersecting size of~$\DD$.

\subsection{Cutting a surface along a drawing}\label{sec:cutting}

Let $D$ be a drawing of a graph~$G$ on~$\surf$.  We first need to define what we mean by \emph{cutting}~$\surf$ along (the relative interior of the edges of)~$D$ (see Figure~\ref{fig:reduc-complex}(a--b)).  Intuitively, we cut $\surf$ along the relative interior of edges of~$D$, resulting in a topological space that is actually homeomorphic to a 2-complex.  Note that each point of~$\surf$ that is an endpoint of some edge of~$D$ still corresponds to a single point on the 2-complex obtained by cutting along~$D$.  More formally, recall from \Cref{sec:storing} that $D$ is represented by a pair $(T,W)$ where $T$ is a triangulation of~$\surf$ and $W$ is a set of walks in the triangulation and in the interior of the surface, such that each edge of~$T$ is used at most once.  The triangulation~$T$ is obtained from a collection of initially disjoint triangles by gluing some vertices and edges together.  Cutting $\surf$ along~$D$ is best described by starting with the same collection of disjoint triangles and specifying a \emph{subset} of these gluing operations:
\begin{itemize}
  \item whenever two directed edges $e$ and~$e'$ of triangles are identified to a single edge in~$T$, and that edge is not used by any walk in~$W$, we identify $e$ and~$e'$.  Note that this process automatically identifies the source endpoints of~$e$ and~$e'$, and similarly for their target endpoints;
  \item whenever two vertices of the resulting triangulation correspond to the same vertex of~$T$, and that vertex is an endpoint of at least one walk in~$W$, we identify these two vertices.
\end{itemize}
Since we only perform a subset of the identifications of vertices and edges in~$T$, the resulting space is naturally a 2-complex.

\subsection{Overview of the reduction}

The proof of \Cref{thm:main} is a parameterized Turing reduction to the embeddability problem on 2-complexes.  
Consider an instance $(G,p,q_1,\ldots,q_c)$ of the \DCREM{}.  We define an uncolored graph $G_2=G_2(G,p,q_1,\ldots,q_c)$ and a set of 2-complexes $\Gamma=\Gamma(p,q_1,\ldots,q_c)$ such that our \DCREM{} instance is positive if and only if $G_2$ embeds in at least one of the 2-complexes in~$\Gamma$ (see a sketch in \Cref{fig:reduc-complex}).  The set~$\Gamma$ is built by branching over a (small enough) set of properties for the hypothetical drawing of a graph~$G'$ (obtained from~$G$ by edge removals and vertex splits as in the problem definition) or, equivalently, by \emph{guessing} some properties of that drawing.

A key definition is the following.  A \emph{fully intersecting} drawing is a drawing without any isolated vertex in which every edge is involved in at least one intersection point (either a self-intersection or an intersection with another edge, but common endpoints do not count).  If a fully intersecting drawing~$D^\times$ has intersecting size at most~$r$, then it is made of at most~$2r$ pieces, because each piece is incident to at least one intersection point, and each intersection point of multiplicity~$k$ is incident to at most $2k$ pieces.

Assume that~$G'$ has a colored drawing~$D'$ in~$\DD$.  Let $D^\times$ be the subdrawing of the subgraph of~$G'$ made of the edges that are involved in at least one intersection in~$D'$, together with their endpoints; it is a fully intersecting drawing, so it has at most $2r$ pieces.  Viewing $D^\times$ as an abstract drawing, without its correspondence with the vertices and edges of~$G'$, this implies that we can enumerate all such colored drawings up to self-homeomorphism of~$\surf$ in time $(c+s+r)^{\OO(s+r)}$.  In other words, we can \emph{guess} the appropriate colored drawing~$D^\times$.

\begin{figure}%
    \def\svgwidth{\linewidth}
    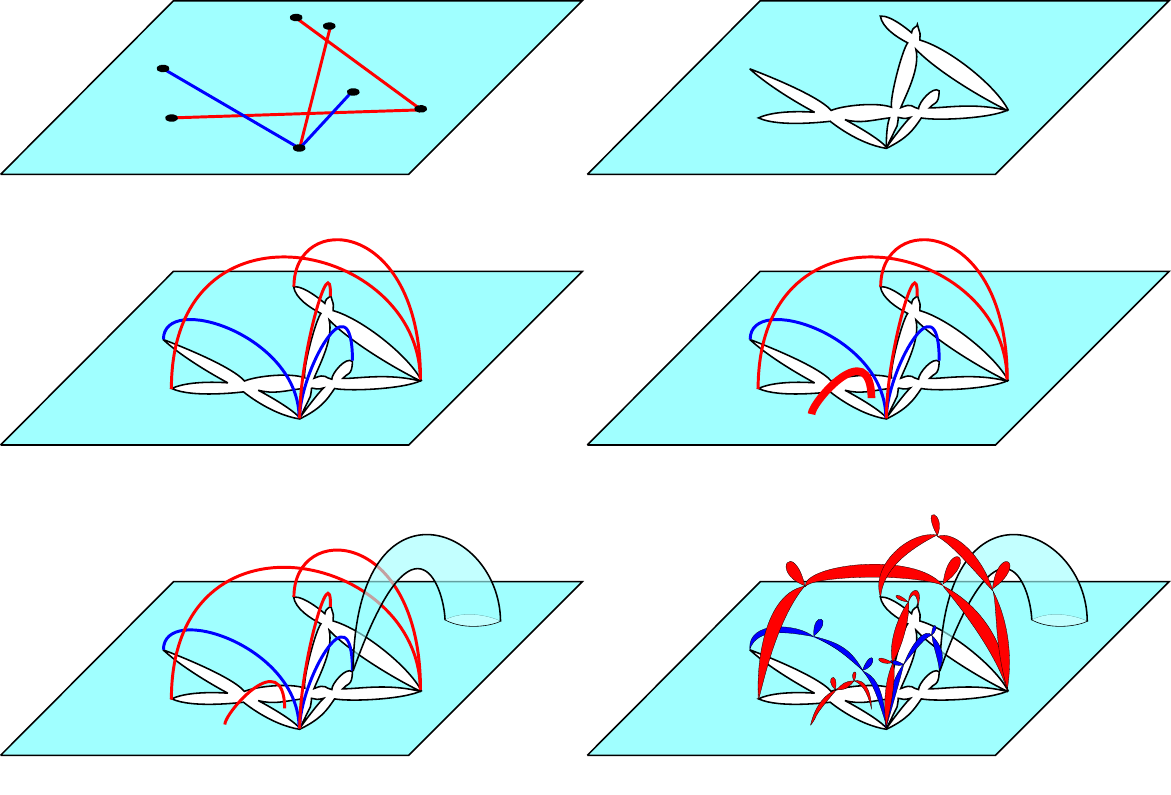
    \caption{The construction of a 2-complex~$\complex_5$ in our reduction.  (a): The surface~$\surf$ and the colored drawing~$D^\times$.  (b): $\complex_1$ is obtained by cutting~$\surf$ along~$D^\times$.  Note that each vertex of~$G$ (in particular, the one in the bottom of the image) on~$\surf$ corresponds to a single point in~$\complex_1$, while each intersection point in~$D^\times$ of multiplicity~$i$ corresponds to~$2i$ points in~$\complex_1$.  (c) $\complex_2$ is obtained by adding isolated edges corresponding to the edges of~$D^\times$; we define the color of each isolated edge to be that of the corresponding edge in~$D^\times$.  (d) $\complex_3$ is obtained by adding $q_i$ isolated edges of color~$i$ for each color~$i$.  Here color~1 is represented in red, color~2 is represented in blue, and $q_1=1$, $q_2=0$.  (e) $\complex_4$ is obtained by iteratively identifying $p$ pairs of points of~$\complex_3$ (here $p=1$).  (f) $\complex_5$ is obtained by replacing each isolated edge of color~$i$ with a necklace (represented figuratively here) of thickness $3r+p+3q+i$ and beads of size $c-i+2$.}
    \label{fig:reduc-complex}
\end{figure}

Subsequently, we cut~$\surf$ along the relative interior of the edges of~$D^\times$ (\Cref{fig:reduc-complex}(b)) and add isolated edges connecting the endpoints of the edges in~$D^\times$ (\Cref{fig:reduc-complex}(c)), obtaining a 2-complex~$\complex_2$ in which $G'$~embeds.  
If $G'$ is allowed to result from~$G$ by edge removals and vertex splits, we modify the 2-complex appropriately, by guessing where the endpoints of the additional isolated edges must be inserted  (\Cref{fig:reduc-complex}(d)) and which points of the 2-complex must be identified back (doing an inverse of the vertex split, \Cref{fig:reduc-complex}(e)).  
Note that these guessed points may be picked from the surface part of $\complex_2$, as well as among the vertices of the edges of $D^\times$ in $\complex_2$.
We obtain a 2-complex~$\complex_4$ in which $G$~embeds.  Each isolated edge of~$\complex_4$ naturally bears
the color of the edge of~$G$ it carries.

\smallskip
Conversely, we would like that an embedding of~$G$ into one of the resulting 2-complexes implies a positive \DCREM{} instance.  We do not exactly achieve this, but need to modify $G$ and the 2-complexes, as follows.
First (\Cref{fig:reduc-graph}(b)), in a minor preprocessing step, we attach a 4-clique to each vertex of~$G$, to ensure that all vertices of~$G$ have degree at least three.
Let $G_1$ be this new graph.  Second, and more importantly  (\Cref{fig:reduc-complex}(f) and~\Cref{fig:reduc-graph}), we need to encode the colors of the edges.
We turn each edge of~$G_1$ and each isolated edge of the $2$-complex~$\complex_4$ (in which $G_1$ embeds, see above) into a \emph{necklace} encoding the color, such that a necklace of a certain encoding type in the resulting graph can only use a necklace of the same type in the resulting 2-complex.  This results in an (uncolored) graph~$G_2$ embedded in a 2-complex~$\complex_5$.  Finally, the set~$\Gamma$ of 2-complexes in which we try to embed~$G_2$ is made of all 2-complexes~$\complex_5$, over all possible choices (guesses) described above.

In more detail, we define a \emph{necklace of thickness~$h$ and beads of size~$k$} to be the graph obtained from a path of length three by (1) replacing each edge with $h$ parallel edges, and (2) attaching $k$ loops to each of the two internal vertices of the path (\Cref{fig:reduc-graph}(c)).
The intuition is that, by the minimum-degree condition on $G_1$, the internal vertices of necklaces of our 2-complex can essentially only be used by internal vertices of the necklaces of $G_2$. So, if edges of color $i$ are encoded with necklaces of, say, thickness $i$ and beads of size $c-i$, we effectively prevent graph necklaces from using 2-complex necklaces~of~different~types.
(We will actually use slightly refined formulas for the thickness and the size of the beads.)

\begin{figure}
    \def\svgwidth{\linewidth}
\begingroup%
  \makeatletter%
  \providecommand\color[2][]{%
    \errmessage{(Inkscape) Color is used for the text in Inkscape, but the package 'color.sty' is not loaded}%
    \renewcommand\color[2][]{}%
  }%
  \providecommand\transparent[1]{%
    \errmessage{(Inkscape) Transparency is used (non-zero) for the text in Inkscape, but the package 'transparent.sty' is not loaded}%
    \renewcommand\transparent[1]{}%
  }%
  \providecommand\rotatebox[2]{#2}%
  \newcommand*\fsize{\dimexpr\f@size pt\relax}%
  \newcommand*\lineheight[1]{\fontsize{\fsize}{#1\fsize}\selectfont}%
  \ifx\svgwidth\undefined%
    \setlength{\unitlength}{486.4864137bp}%
    \ifx\svgscale\undefined%
      \relax%
    \else%
      \setlength{\unitlength}{\unitlength * \real{\svgscale}}%
    \fi%
  \else%
    \setlength{\unitlength}{\svgwidth}%
  \fi%
  \global\let\svgwidth\undefined%
  \global\let\svgscale\undefined%
  \makeatother%
  \begin{picture}(1,0.66956993)%
    \lineheight{1}%
    \setlength\tabcolsep{0pt}%
    \put(0.60481583,0.01178067){\color[rgb]{0,0,0}\makebox(0,0)[lt]{\lineheight{1.49999988}\smash{\begin{tabular}[t]{l}(c)\end{tabular}}}}%
    \put(0.44484248,0.36671368){\color[rgb]{0,0,0}\makebox(0,0)[lt]{\lineheight{1.49999988}\smash{\begin{tabular}[t]{l}$G_2$\end{tabular}}}}%
    \put(0.07302556,0.56557161){\color[rgb]{0,0,0}\makebox(0,0)[lt]{\lineheight{1.49999988}\smash{\begin{tabular}[t]{l}$G$\end{tabular}}}}%
    \put(0.06942132,0.20183468){\color[rgb]{0,0,0}\makebox(0,0)[lt]{\lineheight{1.49999988}\smash{\begin{tabular}[t]{l}$G_1$\end{tabular}}}}%
    \put(0.13383792,0.02036422){\color[rgb]{0,0,0}\makebox(0,0)[lt]{\lineheight{1.49999988}\smash{\begin{tabular}[t]{l}(b)\end{tabular}}}}%
    \put(0.13428958,0.43035123){\color[rgb]{0,0,0}\makebox(0,0)[lt]{\lineheight{1.49999988}\smash{\begin{tabular}[t]{l}(a)\end{tabular}}}}%
    \put(0,0){\includegraphics[width=\unitlength,page=1]{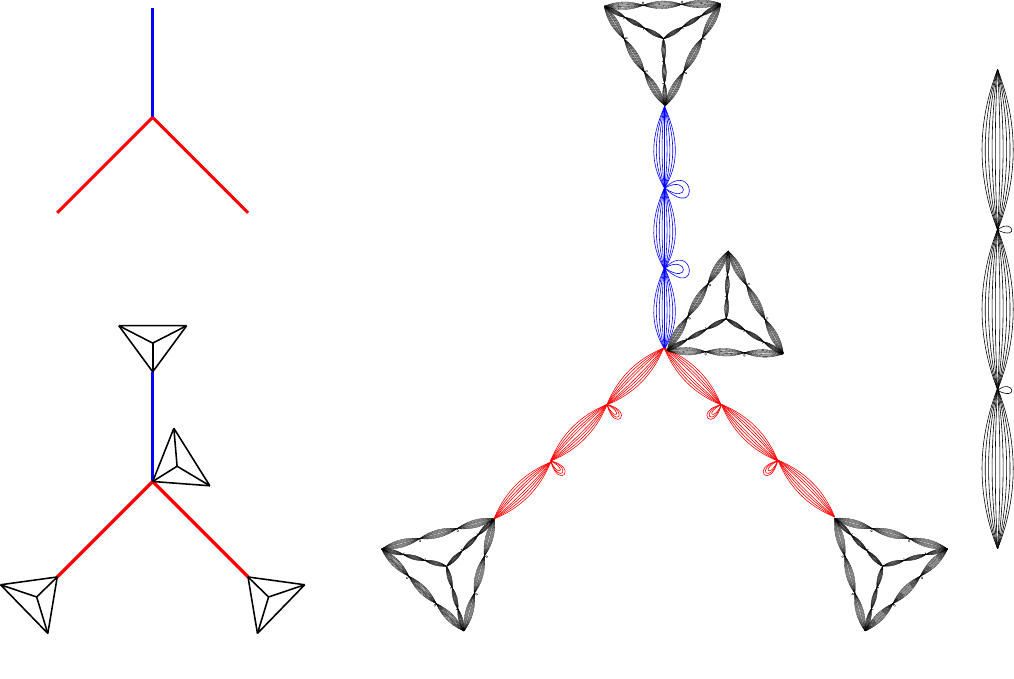}}%
  \end{picture}%
\endgroup%

    \caption{The construction of the graph~$G_2$.  (a): The graph~$G$; the red edges have color~1 and the blue edge has color~2.  (b): The graph~$G_1$.  The new edges, in black, have color~3.  (c): The graph~$G_2$ in the case $c=2$, $r=p=q=1$.  Each edge of color~$i$ is replaced with a necklace of thickness $3r+p+3q+i=7+i$ and beads of size $c-i+2=4-i$.  On the right,  a close-up on the replacement for a black edge in~$G_1$, which is a necklace of thickness ten and beads of size one.}
    \label{fig:reduc-graph}
\end{figure}

\subsection{Formal description of the reduction}\label{sec:formal-reduc}

We now define our reduction precisely.  We first describe the uncolored graph~$G_2=G_2(G,p,q_1,\ldots,q_c)$, which will be the common input to all embeddability instances.  See Figure~\ref{fig:reduc-graph}.
\begin{itemize}
    \item Let $G_1$ be the colored graph obtained from~$G$ by attaching a 4-clique to each vertex~$v$ of~$G$, that is, by adding three new vertices making a 4-clique with~$v$. The edges of each attached 4-clique get the new color $c+1$.
    \item Let $G_2$ be the uncolored graph obtained from~$G_1$ by replacing each edge of color~$i$ with a necklace of thickness~$3r+p+3q+i$, where $q=\sum_{j=1}^cq_j$, and beads of size~$c-i+2$.
\end{itemize}

Let us now describe our set~$\Gamma$ of 2-complexes in which we try to embed~$G_2$.  See Figure~\ref{fig:reduc-complex}.  These 2-complexes depend on several successive choices.  First, let~$D^\times\in\DD$ be a fully intersecting colored drawing on~$\surf$ (thus with colors in $\{1,\ldots,c\}$).
\begin{itemize}
    \item Let $\complex_1$ be obtained from~$\surf$ by cutting~$\surf$ along the relative interior of each edge of~$D^\times$.
    \item Let $\complex_2$ be obtained from~$\complex_1$ by doing the following: For each edge of~$D^\times$ of color~$i$ ($i=1,\ldots,c$), connecting points $x$ and~$y$ of~$\complex_1$, we add an isolated edge connecting $x$ and~$y$ to the 2-complex, which is said to have \emph{color}~$i$.  (This may require subdividing~$\complex_1$ to turn $x$ and~$y$ into vertices.)
\end{itemize}

Let $U$ be the set of points of~$\complex_2$ that correspond to a single point of the interior of~$\surf$.  In detail, $U$ is made of (1) the points of~$\complex_2$ that correspond to an endpoint of an edge in~$D^\times$, and (2) the interior of the surface part of~$\complex_2$.

For each $i=1,\ldots,c$, let $(x_i^1,y_i^1),\ldots,(x_i^{q_i},y_i^{q_i})$ be pairs of points on~$U$.  (These points are not necessarily distinct.  There are finitely many choices up to a self-homeomorphism of~$\surf$.)
\begin{itemize}
    \item Let $\complex_3$ be obtained from~$\complex_2$ by doing the following: For each $i=1,\ldots,c$ and each $j=1,\ldots,q_i$, we add to the 2-complex an isolated edge connecting $x_i^j$ and~$y_i^j$, which is said to have \emph{color}~$i$.
\end{itemize}
For each unordered partition $p=p_1+\ldots+p_k$ of the integer~$p$, where $p_1,\ldots,p_k$ are positive integers, and for each $i=1,\ldots,k$, let $z_i^1,\ldots,z_i^{p_i+1}$ be $p_i+1$ points on $U$ (which is also a part of~$\complex_3$), such that all these $p+k$ points are pairwise distinct.  (Note that the points $z_i^j$ are not necessarily distinct from the points $x_{i'}^{j'}$ and $y_{i'}^{j'}$.)
\begin{itemize}
    \item Let $\complex_4$ be obtained from~$\complex_3$ by doing the following: For each $i=1,\ldots,k$, we identify the points $z_i^1,\ldots,z_i^{p_i+1}$ to a single point~$z_i$.
    \item Let $\complex_5$ be obtained from~$\complex_4$ by replacing each isolated edge~$e$ of~$\complex_1$ of color~$i$ ($i=1,\ldots,c$) with a necklace of thickness~$3r+p+3q+i$ and beads of size $c-i+2$.
\end{itemize}
To conclude, $\complex_5$ depends on the choice of the drawing~$D^\times\in\DD$, of the points $x_i^j$, $y_i^j$, of the partition of~$p$, and of the points~$z_i^j$.  We let $\Gamma=\Gamma(p,q_1,\ldots,q_c)$ be the set of 2-complexes~$\complex_5$ that can be obtained in this way, over all these choices.

\subsection{Correctness of the reduction}\label{sec:correctness-meta}

The correctness of the reduction is established in the two following lemmas.

\begin{lemma}\label{lem:drawing-to-embedding}
  Assume that there is a drawing~$D'$ in~$\DD$ of a colored graph~$G'$ obtained from~$G$ by performing at most $p$ successive vertex splits and by removing, for $i=1,\ldots,c$, at most $q_i$ edges of color~$i$.  Then $G_2(G,p,q_1,\ldots,q_c)$ embeds in some 2-complex in $\Gamma(p,q_1,\ldots,q_c)$.
\end{lemma}
\begin{proof}
  Let $D^\times$ be the colored subdrawing of~$D'$ made of the edges that carry at least one intersection in~$D'$, together with their endpoints.  Then $D^\times$ is fully intersecting, and it belongs to~$\DD$, because $\DD$ is stable.  Let $\complex_1$ and~$\complex_2$ be defined as above, for that specific choice of~$D^\times\in\DD$.  By construction, $G'$ embeds into~$\complex_2$ in such a way that each isolated edge of~$\complex_2$ of color~$i$ coincides with an edge of~$G'$ of color~$i$.

  It follows that, for an appropriate choice of the partition of~$p$ and of the vertices $x_i^j,y_i^j,z_i^j$, the graph~$G$ embeds into $\complex_4$ (again, we reuse the notations introduced in the description of the reduction), in such a way that each isolated edge of~$\complex_4$ of color~$i$ coincides with an edge of~$G$ of color~$i$ or is not used by the embedding.
  
  Note that, by construction, in that embedding the vertices of~$G$ are mapped to the surface part of~$\complex_4$, and thus we can attach a 4-clique to each vertex of~$G$ while preserving the fact that we have an embedding.  Thus, the graph~$G_1$ embeds into~$\complex_4$ too, in such a way that each isolated edge of~$\complex_4$ of color~$i$ ($i=1,\ldots,c$) coincides with an edge of~$G$ of color~$i$ or is not used at all.
 
  Finally, replacing each isolated edge of~$\complex_4$ of color~$i$, and similarly each edge of~$G_1$ of color~$i$, with a necklace of thickness $3r+p+3q+i$, where $q=\sum_{j=1}^cq_j$, and beads of size $c-i+2$ implies that $G_2$ embeds into~$\complex_5$.  Indeed, each edge of~$G_1$ that coincided with an isolated edge of~$\complex_4$ is replaced with a necklace that can be embedded in the corresponding necklace of~$\complex_5$; and each edge of~$G_1$ that belonged to the surface part of~$\complex_4$ is replaced with a necklace that can be drawn in a neighborhood of the edge of~$G_1$ without intersecting the rest of the drawing.  This concludes the proof.
\end{proof}

Conversely:
\begin{lemma}\label{lem:embedding-to-drawing}
  Assume that the graph~$G_2(G,p,q_1,\ldots,q_c)$ embeds in some 2-complex in $\Gamma(p,q_1,\ldots,q_c)$.  Then, there exists a drawing~$D'$ in~$\DD$ of a colored graph~$G'$ obtained from~$G$ by performing at most $p$ successive vertex splits and by removing, for $i=1,\ldots,c$, at most $q_i$ edges of color~$i$.
\end{lemma}
\begin{proof}
  Let $G_2:=G_2(G,p,q_1,\ldots,q_c)$, and let $\complex_5$ be some 2-complex in~$\Gamma(p,q_1,\ldots,q_c)$ in which $G_2$ can be embedded.  The 2-complex~$\complex_5$ is defined by the choice of a fully intersecting drawing~$D^\times\in\DD$, of pairs $(x_i^j,y_i^j)$, of a partition of~$p$, and of the points~$z_i^j$.  We also freely use the notations $G_1, \complex_1,\ldots,\complex_5$ used in the description of the reduction.

  First, in a series of claims, we prove that we can assume that each necklace of~$\complex_5$ is either not used at all by~$G_2$, or used by a necklace of~$G_2$ of the same type.
  \begin{itemize}
  \item Let $v$ be a vertex of~$G$.  We claim that, in the embedding of~$G_2$ in~$\complex_5$, $v$ cannot be mapped to the interior of a necklace~$N$ of~$\complex_5$.  Indeed, otherwise, since $v$ has degree at least three in~$G_2$, it is mapped to some interior vertex of~$N$ (i.e., not to an interior point of an edge of~$N$).  Moreover, $v$, as a vertex in~$G_2$, has at least three distinct neighbors, each of degree at least three, and thus mapped either to a vertex of~$N$, or to the complement of~$N$.  But, starting from the location of~$v$ in~$N$, there do not exist three disjoint paths going to different vertices of~$N$ or to the complement of~$N$.  This contradiction proves the claim.

  \item Assume now that some vertex~$v$ of~$G_2$ is mapped to the interior of a necklace~$N$ of~$\complex_5$, of thickness $3r+p+3q+i$ (where, again, $q=\sum_{j=1}^cq_j$), and beads of size $c-i+2$, for some~$i$.  We claim that $N$ actually contains a necklace of~$G_2$ of the same type.  Indeed, by the previous paragraph, $v$ is not a vertex of~$G$, and is thus an interior vertex of a necklace of~$G_2$.  Moreover, $v$ has degree at least three, so it is mapped to an interior vertex~$x$ of~$N$.  This vertex~$x$ of~$N$ is connected to each of its two distinct neighbors by $3r+p+3q+i$ parallel edges, and is incident to $c-i+2$ loops.  Moreover, $v$ is part of a necklace of thickness $3r+p+3q+j$ and beads of size $c-j+2$, for some~$j$.  Considerations similar to those of the previous paragraph imply that $i=j$.  Arguing similarly with the neighbors of~$v$ that are not vertices of~$G$ yields the result.

  \item We claim that without loss of generality, if the endpoints of an edge~$e$ of~$G_2$ avoid the interior of all necklaces of~$\complex_5$, then $e$ avoids all necklaces of~$\complex_5$.  To prove this, there are two cases:
  \begin{itemize}
      \item Assume that $e$ is a loop.  Let $v$ be the vertex of~$G_2$ incident to~$e$.  Since $v$ avoids the interior of all necklaces of~$\complex_5$, the loop~$e$ can be redrawn to a small neighborhood of~$v$, while preserving the fact that we have an embedding of~$G_2$ into~$\complex_5$.
      \item Assume that $e$ is not a loop.  It is part of a batch of at least $3r+p+3q+1$ parallel edges.  The interior of each necklace of~$\complex_5$ can be visited by at most one such edge, because the endpoints of~$e$ are outside it; also, each singular point of~$\complex_5$ can be visited by at most one such edge.  Moreover, $\complex_5$ has at most $r+q$ necklaces, and at most $2r+p+2q$ singular points that avoid the interior of all necklaces.  It follows that for some edge~$e'$ in the batch, the relative interior of~$e'$ avoids the interiors of all necklaces and all singular points.  We can thus reroute $e$ parallel to~$e'$, avoiding all necklaces, while preserving the fact that we have an embedding of~$G_2$ into~$\complex_5$.
  \end{itemize}
  \end{itemize}

  As a summary of the above considerations, we now have an embedding of~$G_2$ into~$\complex_5$ in such a way that each necklace of~$\complex_5$ is either not used, or used by a necklace of~$G_2$ of the same type.  This implies that $G_1$ embeds in~$\complex_4$ in such a way that each isolated edge of~$\complex_4$ of color~$i$ ($i=1,\ldots,c$) coincides with an edge of~$G_1$ of color~$i$, or is not used at all.  The same holds for~$G$, since $G$ is a subgraph of~$G_1$.

  Recall that $\complex_4$ is obtained by identifying points of~$\complex_3$ into $k$~points $z_1,\ldots, z_k$.  Let $S$ be the set of indices $i$ such that $z_i$ contains a point in the relative interior of an edge of~$G$.  We do the following:
  \begin{itemize}
  \item While, for some $i\in S$, $z_i$ contains a point of the relative interior of an edge $e=uv$ of~$G$, we do the following.  First, we perform a vertex split of one of its vertices, say~$v$, by replacing $uv$ with $uv'$, where $v'$ is a new vertex.  Second, we shrink edge $uv'$ without moving~$u$, so that in the new embedding of~$G$, edge $uv'$ does not contain any singular point of~$\complex_4$ in its relative interior.  The number of vertex splits performed so far is at most the size of~$S$, and after this step, the image of~$G$ avoids the points $z_i$, $i\in S$.
  \item Now, for each $i=1,\ldots,k$, we split $z_i$ into $p_i+1$ points in order to obtain the 2-complex~$\complex_3$.  In passing, if vertex~$z_i$ is used by a vertex~$v$ of~$G$, we perform at most $p_i$ splits of~$G$ at~$v$ in order to preserve the property that we have an embedding of some graph.
  \end{itemize}
  We now have an embedding, into~$\complex_3$, of a graph~$G''$ obtained from~$G$ by at most $|S|+\sum_{i\in S}p_i\le p$ vertex splits.  Moreover, as above, each isolated edge of $\complex_3$ of color~$i$ coincides with an edge of~$G''$ of color~$i$, or is not used at all.

  If we remove the edges embedded in $\complex_3\setminus\complex_2$, we obtain an embedding, in~$\complex_2$, of a graph~$G'$ obtained from~$G$ by at most $p$ vertex splits and by removing, for each~$i$, at most $q_i$ edges of color~$i$.  Moreover, again, each isolated edge of $\complex_2$ of color~$i$ coincides with an edge of~$G'$ of color~$i$, or is not used at all.

  Finally, we obtain a drawing~$D'$ of~$G'$ into~$\surf$ that results from~$D^\times$ by removing vertices and edges, and by adding vertices and edges that do not intersect the rest of the drawing.  Thus, $D'\in\DD$.
\end{proof}

\subsection{End of proof}\label{sec:end-meta}

The two preceding lemmas show that a \DCREM{} instance $(G,p,q_1,\ldots,q_c)$ is positive if and only if $G_2(G,p,q_1,\ldots,q_c)$ embeds in at least one of the 2-complexes in~$\Gamma(p,q_1,\ldots,q_c)$.  To conclude the proof, we show that the reduction described above turns into an actual algorithm, and analyze its complexity.  We first state and prove the result announced in Section~\ref{sec:storing} on the size of a representation of a drawing:

\begin{lemma}\label{lem:represent-drawings}
  Let $D$ be a drawing of a graph on a surface~$\surf$.  Then $D$ is represented by some pair~$(T,W)$ in which $T$ has at most $96(s+u)$ triangles, where $s$ is the topological size of~$\surf$, and the intersection points of~$D$ subdivide the edges of~$D$ into $u$ pieces in total.
\end{lemma}
\begin{proof}%
  We view $D$ as a graph~$G$ embedded in the interior of~$\surf$.  The vertices of~$G$ are the endpoints of the edges of~$D$ and the intersection points of~$D$; the edges of~$G$ are the pieces of the edges of~$D$.  A \emph{weak triangulation} is a graph~$T$ embedded on~$\surf$ such that each face of~$T$ (each connected component of the complement of the image of~$T$) is homeomorphic to a disk bounded by three edges, but this time some incident vertices and/or edges may coincide; in other words, a weak triangulation is not necessarily a 2-complex.  We gradually augment~$G$ to a weak triangulation of~$\surf$.  Recall that the topological size of~$\surf$ equals $s=d+g+b$, where $d$ is the number of connected components, $g$ is the genus, and $b$ is the number of boundary components of~$\surf$.  In a first case, we assume that $\surf$ is connected and intersects~$D$.

  We first cover each of the $b$~boundary components of~$\surf$ with a loop.  Now, consider a face~$f$ of this new graph (still denoted~$G$).  It has at least one boundary component, and each boundary component contains at least one vertex.  Recall that $G$ has no isolated vertex; let  $\ell'\ge1$ be the number of edges on the boundary of~$f$.  We distinguish several cases:
  \begin{itemize}
  \item If $f$ is homeomorphic to a disk and $\ell'=1$, then we insert a new vertex in the interior of the disk, and connect it to the unique vertex on the boundary of~$f$, turning~$f$ into $\ell'=1$ triangle;
  \item if $f$ is homeomorphic to a disk and $\ell'=2$, then we insert a new vertex in the interior of the disk, and connect it to each of the two vertices on the boundary of~$f$, turning~$f$ into $\ell'=2$ triangles;
  \item otherwise, let $g'\ge0$ and~$b'\ge1$ be the genus and number of boundary components of~$f$.  We now summarize a standard argument to triangulate~$f$ without adding vertices: One can add $g'$~loops to decrease the genus of~$f$ to zero, and there are still $b'$ boundary components, now with $\ell'+2g'$ sides in total.  Adding $b'-1$ edges, we connect all the boundary components of~$f$ into a single one; $f$ is now a disk with $\ell'+2g'+2b'-2$ sides, which is at least three (otherwise $f$ would be of one of the first two types).  In turn, this disk can be (weakly) triangulated into $\ell'+2g'+2b'-4$ triangles.
  \end{itemize}
  We now have a weak triangulation of~$f$ with at most $\ell'+2g'+2b'$ triangles. Initially, $G$ has at most $2u$ faces, and moreover, $\sum_fg'\le g$, $\sum_fb'\le b+2u$, and $\sum_f\ell'\le2u+b$.  So, to conclude the first case where $\surf$ is connected and intersects~$D$, we obtain that $G$ is weakly triangulated using at most $2u+b+2g+2b+4u=6u+2g+3b$ triangles.
  
  In the second case, if $\surf$ is still connected but does not intersect~$D$, we just insert a single vertex and a loop in~$G$ and apply the above process, weakly triangulating $\surf$ with at most $6+2g+3b$ triangles.

  So, if $\surf$ is connected, we can weakly triangulate it using at most $6+2g+3b+6u$ triangles.  To turn this weak triangulation into a triangulation, it suffices to apply barycentric subdivisions twice; the number of triangles is multiplied by~16.

  Finally, we obtain the result by summing up over all connected components.
\end{proof}

Recall that we consider drawings up to homeomorphism of~$\surf$.  The next lemma explains how to enumerate (representations of) the drawings~$D^\times$.
\begin{lemma}\label{lem:enum}
  Let $\DD^\times$ be the class of fully intersecting colored drawings that are in~$\DD$.  In $(c+s+r)^{\OO(s+r)}\cdot\delta(192(s+r)c)$ time, we can enumerate a set of $(c+s+r)^{\OO(s+r)}$ representations of drawings in~$\DD^\times$, each of size $\OO((s+r)c)$, such that each drawing in~$\DD^\times$ is represented at least once.
\end{lemma}
\begin{proof}
  By definition of~$r$, any drawing~$D^\times\in\DD^\times$ has at most $2r$ pieces, and can thus be represented by a pair $(T,W)$ such that $T$ has at most $192(s+r)$ triangles, by Lemma~\ref{lem:represent-drawings}.

  Then, starting with at most $192(s+r)$ triangles, we identify some of their sides in pairs, obtaining a triangulation~$T$; there are $(s+r)^{\OO(s+r)}$ possibilities.  For each such choice, we compute the genus, orientability, and number of boundary components of each connected component of the resulting surface using standard methods (in particular the Euler characteristic), and discard~$T$ if it is not homeomorphic to~$\surf$, all in time polynomial in $s+r$.

  Let $T$ be one of the remaining triangulations.  We choose an arbitrary set of at most~$2r$ walks~$W$ such that each edge of~$T$ is used at most once by the union of these walks; there are $(s+r)^{\OO(s+r)}$ possibilities.  We also assign a color to each of these walks, adding a factor of $\OO(c^{2r})$ in the number of choices.  We discard the choices in which some walk has no intersection with any other walk.

  Finally, for each of the $(c+s+r)^{\OO(s+r)}$ resulting representations of colored drawings~$D^\times$, we decide whether $D^\times\in\DD$ in $\delta(192(s+r)c)$ time, and discard those not in~$\DD$.
\end{proof}

\begin{proof}[Proof of Theorem~\ref{thm:main}]
  The algorithm tests the embeddability of $G_2=G_2(G,p,q_1,\ldots,q_c)$ into each 2-complex in~$\Gamma=\Gamma(p,q_1,\ldots,q_c)$, and returns ``yes'' if and only if at least one embeddability test succeeded.  The correctness of the reduction follows from Lemmas~\ref{lem:drawing-to-embedding} and~\ref{lem:embedding-to-drawing}.  There remains to provide some details on the algorithm and the runtime.

  $G_2$ has size $\OO((c+r+p+q)n)$ and can be computed in $\OO((c+r+p+q)n)$ time.
  
  Using Lemma~\ref{lem:enum}, we can compute a set of $(c+s+r)^{\OO(s+r)}$ representations of the drawings~$D^\times$ appearing in the definition of the set of 2-complexes~$\Gamma$ in $(c+s+r)^{\OO(s+r)}\cdot\delta(192(s+r)c)$ time, and each representation has size $\OO((s+r)c)$.  Given such a representation, we can compute $\complex_1$ and then $\complex_2$ in time polynomial in $c+s+r$, each of size $\OO(s+r)$ (ignoring the colors on the isolated edges).

  The number of partitions of the integer~$p$ is $2^{\OO(p)}$.  We select the points $x_i^j$, $y_i^j$, and~$z_i^j$ on~$U$.  There are at most~$2r$ points in~$U$ that correspond to an endpoint of an edge in~$D^\times$; moreover $D^\times$ has at most $s+2r$ faces in~$\surf$, and any two points in the same face are equivalent (up to a self-homeomorphism of~$\surf$) as far as the choice of the point is concerned.  There are at most $2p+2q$ points, and thus at most $(s+4r+2p+2q)^{2p+2q}$ choices in total.  Thus, for a given~$D^\times$, there are $(s+r+p+q)^{\OO(p+q)}$ possible choices of 2-complexes~$\complex_5$.  So, in total, there are $(c+s+r+p+q)^{\OO(c+s+r+p+q)}$ choices for~$\complex_5$. 

  To compute the 2-complex corresponding to a given choice, if necessary, we first refine the 2-complex~$\complex_2$ by subdividing triangles so that the points $x_i^j$, $y_i^j$, and~$z_i^j$ become vertices of the new 2-complex.  Then we add the $q$ isolated edges, obtaining~$\complex_3$, and perform the $p$ identifications of vertices, obtaining~$\complex_4$.  Finally, $\complex_5$ is obtained from~$\complex_4$ by replacing the isolated edges with necklaces.  All this takes time polynomial in $c+s+r+p+q$, and the size of the resulting complex~$\complex_5$ is $\OO(s+(c+r+p+q)(r+q))$.

  So there are $(c+s+r+p+q)^{\OO(c+s+r+p+q)}$ 2-complexes in~$\Gamma$, each of size $\OO(s+(c+r+p+q)(r+q))$, and the computation of~$\Gamma$ takes $(c+s+r+p+q)^{\OO(c+s+r+p+q)}\cdot\delta(192(s+r)c)$ time.

  Finally, we test the embeddability of~$G_2$ into each 2-complex in~$\Gamma$ using \Cref{thm:embedding}, which takes $2^{\poly(c+s+r+p+q)}\cdot\delta(192(s+r)c)\cdot n^2$ time, or $f(c+s+r+p+q)\cdot n$ time.
\end{proof}

\section{Applications: FPT algorithms for many crossing number variants}\label{sec:consequences}

In this section, we demonstrate the wide applicability of our framework to crossing number problems.   First, as a toy example, we reprove that we can compute the traditional crossing number in linear FPT time on surfaces~\cite{cm-faeg2-26} (albeit with a worse dependency in the parameter compared to Lokshtanov, Panolan, Saurabh, Sharma, Xue, and Zehavi~\cite{Cr-SODA25} in the plane):
\begin{proposition}
  Let $\surf$ be a surface with topological size~$s$, let $t\geq0$ be an integer, and let $G$ be a graph of size~$n$.   Deciding whether $G$ has crossing number at most~$t$ in $\surf$ can be done in time $2^{\poly(s+t)}\cdot n^2$, or in time $f(s+t)\cdot n$, for some function~$f$.
\end{proposition}
\begin{proof}
  Recall from \Cref{obs:normalcr} that we apply our \DC{} framework to the case where $\DD$ is the class of normal drawings on~$\surf$ with at most~$t$ crossings  and a single color ($c=1$).  Then $\DD$ is stable, and membership in~$\DD$ can be decided in time~$\delta(j)=j^{\OO(1)}$, a polynomial in the size~$j$ of the representation of the input drawing.  Moreover, $r:=2t$ is an upper bound on the intersecting size of~$\DD$, because the only intersection points in a normal drawing are crossings, which have multiplicity two.  \Cref{thm:main} implies the result.
\end{proof}

We now turn to the promised applications.  We survey, in order, specific drawing styles, colored crossing number problems, various counting methods, and problems where prior edge removals and vertex splits are allowed.  Throughout this section, again, let $\surf$ be a surface and $\DD$ the class of colored drawings of graphs on~$\surf$.

In most of our reductions here, we use only the \DCP{} (that is, without vertex splits or edge removals), except in \Cref{sub:delete-split}, where we need the full strength of the \DCREM{}.   
In nearly all cases, only normal drawings of graphs are considered.  Note that in all cases, for positive instances, we can compute a corresponding drawing using \Cref{thm:explicit}; the runtime is that of the quadratic-time algorithm for the corresponding decision problem.

Crossing number variants that prescribe the rotation system of the input graph are postponed to \Cref{sec:metasection-rot}.

\subsection{Variations on the traditional crossing number}

\subparagraph{$k$-planar drawings.}  Introduced by Pach and T\'oth~\cite{DBLP:journals/combinatorica/PachT97}, \emph{$k$-planar drawings} (for $k\ge0$) are {normal} drawings in the plane such that every edge is involved in at most $k$ crossings. The \emph{$k$-planar crossing number} of~$G$ is the smallest number of crossings among all $k$-planar (normal) drawings of~$G$. A graph may have no $k$-planar drawing at all, and then its $k$-planar crossing number is~$\infty$ by convention.
Already for $k=1$, deciding the existence of a $1$-planar drawing is \NP-hard even for very restricted inputs \cite{DBLP:journals/algorithmica/GrigorievB07,DBLP:journals/siamcomp/CabelloM13}.
This definition naturally extends to an arbitrary surface $\surf$, which we call the \emph{$k$-surface crossing number} in~$\surf$. %

\subparagraph{$k$-quasi-planar and min-$k$-planar drawings.}
There are similar concepts of \emph{$k$-quasi-planar}~\cite{DBLP:journals/combinatorica/AgarwalAPPS97} (no $k$ edges pairwise cross) and of \emph{min-$k$-planar}~\cite{DBLP:conf/gd/BinucciBBDDHKLMT23} (for every two edges with a common crossing, one carries at most $k$ crossings) normal drawings in the plane, which give rise to the corresponding crossing number flavors in the plane, and we may likewise generalize them to the \emph{$k$-quasi-surface} and \emph{min-$k$-surface} \emph{crossing numbers} in a surface~$\surf$.

\begin{figure}
    \centering
    \def\svgwidth{.\linewidth}
  \includegraphics[width=0.25\hsize,page=1]{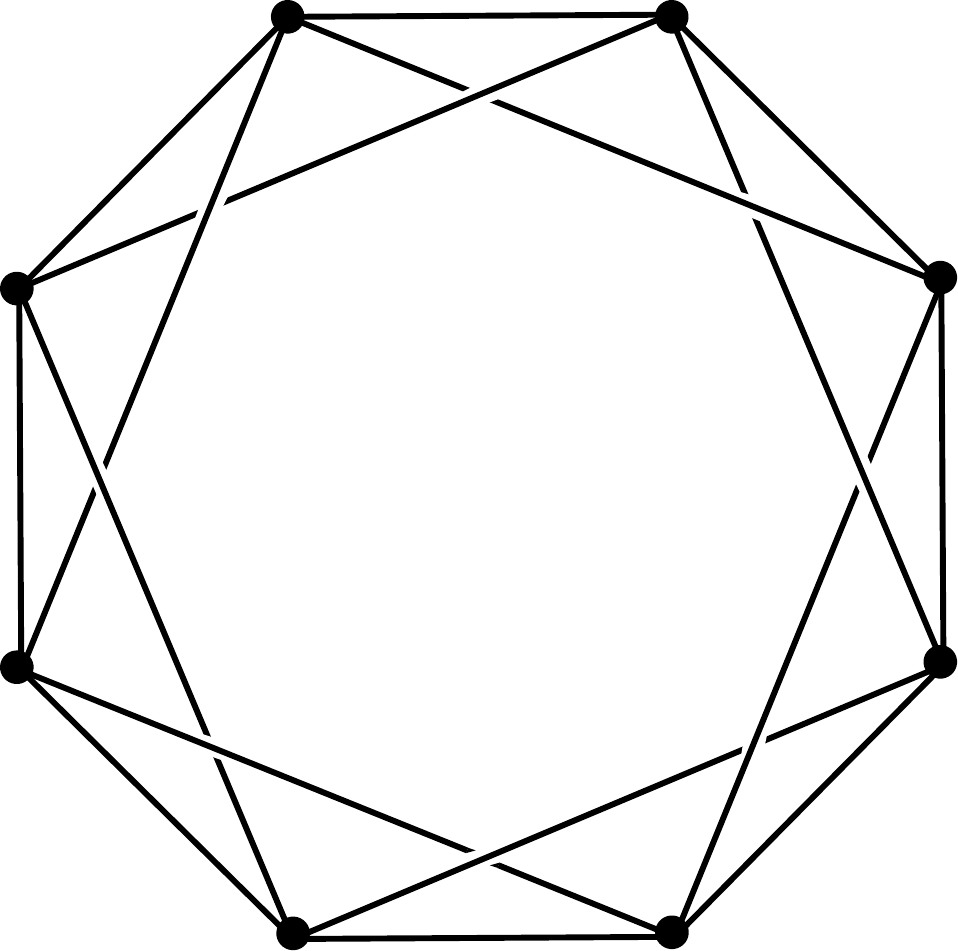}
    \caption{A drawing of some graph that is a $1$-gap drawing, as certified by the casing.}
    \label{fig:gapcr}
\end{figure}

\subparagraph{$k$-gap drawings.}
Another known concept is that of \emph{$k$-gap crossing number} \cite{DBLP:journals/tcs/BaeBCEE0HKMRT18,DBLP:journals/jgaa/BeusekomPS22}, minimizing the number of crossings over \emph{$k$-gap drawings} in the plane (\Cref{fig:gapcr}), namely, normal drawings $D$ of a graph~$G$ admitting a mapping from each crossing in $D$ to one of the two involved edges such that at most $k$ crossings are mapped to each edge of~$G$. This problem is \NP-hard for $k=1$~\cite{DBLP:journals/tcs/BaeBCEE0HKMRT18}. Again, the $k$-gap crossing number is naturally generalized to any surface~$\surf$.

\subparagraph{Restriction to simple drawings.}
Many authors restrict the problems to simple graphs and require the admissible drawings to be simple; this restriction defines the \emph{simple-drawing $k$-surface} (\emph{$k$-quasi-surface}, \emph{min-$k$-surface}, \emph{$k$-gap}) crossing numbers.
These {simple-drawing} variants can be very different from the non-simple ones, see, e.g., \cite[Chapter~7]{schaefer2017crossing} and \cite{DBLP:conf/gd/HlinenyK24}, but we will show that they also smoothly fit into our framework.

\subparagraph{Applying our framework.}
Given $t\geq0$, deciding whether a graph $G$ has $k$-planar ($k$-quasi-planar, or min-$k$-planar) crossing number at most $t$ can be done in quadratic \FPT-time~\cite{patterncrossing24} using Courcelle's theorem, and thus with a huge dependency in the parameter~$t$.  (Note that if $k>t$, then setting $k:=t$ does not change the problem, so we may without loss of generality assume that $k\le t$.)  The parameterized complexity of these three problems in other surfaces, and of the $k$-gap crossing problem altogether, had not been studied prior to us.  Even in the known planar cases, we give two algorithms with a better runtime than previously known:

\begin{proposition}\label{prop:kplanar-etc}%
Let $k\ge1$ and $t\ge0$ be integers, let $\Pi\in\{$`$k$-surface', `$k$-quasi-surface', `min-$k$-surface', `$k$-gap'$\}$, let $\surf$ be a surface with topological size~$s$, and let $G$ be a graph of size~$n$.  Deciding whether $G$ has $\Pi$-crossing number at most~$t$ in $\surf$ can be done in time $2^{\poly(s+t)}\cdot n^2$, or in time $f(s+t)\cdot n$ for some function~$f$.  The same conclusion holds for the simple-drawing $\Pi$-crossing number of a simple input graph~$G$.
\end{proposition}
\begin{proof}
  The proof relies on the algorithm of \Cref{thm:main} for the \DCP{} with a single color ($c=1$).  We let $\DD$ be the class of normal drawings on~$\surf$ satisfying~$\Pi$ and with at most $t$ crossings; it is straightforward to check that $\DD$ is stable.
  
  Given a representation of a drawing $D$, one may decide whether $D\in\DD$ in time polynomial in the size of the representation (and independently of the parameter $k$). This is an easy routine in all cases of $\Pi$ except when $\Pi=$\,`$k$-gap'; in the latter case we additionally employ a standard maximum flow algorithm to decide the existence of a mapping from the crossings to their incident edges respecting ``capacities'' of edges to accept at most $k$ crossings each.
  
  Since we only consider normal drawings, $r:=2t$ is an upper bound on the intersecting size of~$\DD$.  Hence, our conclusion follows by an application of \Cref{thm:main}.  Note that the factor $\delta(s+t)$ is a polynomial in $s+t$ and can be ignored in the runtime, given that there is already a factor of $2^{\poly(s+t)}$.

  Next, we consider the simple-drawing $\Pi$-crossing number, assuming $G$ is simple.  The only modification to the above argument is that we restrict $\DD$ to simple drawings.  We remark that, for this case, it is important that our definition of a simple drawing does not immediately exclude drawings of non-simple graphs; indeed, while this is irrelevant for our simple graph~$G$, adding an uncrossed loop to a drawing~$D\in\DD$ must result in a drawing in~$\DD$.
\end{proof}

\subparagraph{Fan-crossing, fan-planar, and related drawing styles.}
A normal drawing is \emph{fan-crossing}~\cite{DBLP:journals/combinatorics/0001U22} if all edges crossing the same edge are incident to a common vertex.  One can restrict this notion to \emph{strongly} or \emph{weakly fan-planar} drawings, which not necessarily lead to the same crossing number problems, see \Cref{fig:fan} and Cheong, F\"orster, Katheder, Pfister, and Schlipf~\cite{DBLP:conf/gd/CheongFKPS23}.  The \emph{fan-crossing}, \emph{strongly fan-planar}, and \emph{weakly fan-planar crossing numbers} are the minimum number of crossings over all strongly fan-planar, weakly fan-planar, and fan-crossing drawings in the plane, respectively. 
Formally, a fan-crossing drawing is called \emph{weakly fan-planar} if, whenever edges $f_1$ and~$f_2$ with a common end~$v$ both cross an edge~$e$, the (sub)arcs of both $f_1,f_2$ from $v$ to the crossing with~$e$ lie on the same side of~$e$; and it is called \emph{strongly fan-planar} if, moreover, in the described situation of $f_1,f_2$ both crossing $e$, the union $e\cup f_1\cup f_2$ in the drawing does not enclose both ends of~$e$ (in the plane).
We refer to Cheong et al.~\cite{DBLP:conf/gd/CheongFKPS23} for a closer discussion.

\begin{figure}%
    \centering
    \def\svgwidth{\linewidth}
\begingroup%
  \makeatletter%
  \providecommand\color[2][]{%
    \errmessage{(Inkscape) Color is used for the text in Inkscape, but the package 'color.sty' is not loaded}%
    \renewcommand\color[2][]{}%
  }%
  \providecommand\transparent[1]{%
    \errmessage{(Inkscape) Transparency is used (non-zero) for the text in Inkscape, but the package 'transparent.sty' is not loaded}%
    \renewcommand\transparent[1]{}%
  }%
  \providecommand\rotatebox[2]{#2}%
  \newcommand*\fsize{\dimexpr\f@size pt\relax}%
  \newcommand*\lineheight[1]{\fontsize{\fsize}{#1\fsize}\selectfont}%
  \ifx\svgwidth\undefined%
    \setlength{\unitlength}{832.99722467bp}%
    \ifx\svgscale\undefined%
      \relax%
    \else%
      \setlength{\unitlength}{\unitlength * \real{\svgscale}}%
    \fi%
  \else%
    \setlength{\unitlength}{\svgwidth}%
  \fi%
  \global\let\svgwidth\undefined%
  \global\let\svgscale\undefined%
  \makeatother%
  \begin{picture}(1,0.24127389)%
    \lineheight{1}%
    \setlength\tabcolsep{0pt}%
    \put(0.83661058,0.00316534){\color[rgb]{0,0,0}\makebox(0,0)[lt]{\lineheight{1.49999988}\smash{\begin{tabular}[t]{l}(c)\end{tabular}}}}%
    \put(0.41273466,0.00316534){\color[rgb]{0,0,0}\makebox(0,0)[lt]{\lineheight{1.49999988}\smash{\begin{tabular}[t]{l}(b)\end{tabular}}}}%
    \put(0.06652352,0.00316534){\color[rgb]{0,0,0}\makebox(0,0)[lt]{\lineheight{1.49999988}\smash{\begin{tabular}[t]{l}(a)\end{tabular}}}}%
    \put(0,0){\includegraphics[width=\unitlength,page=1]{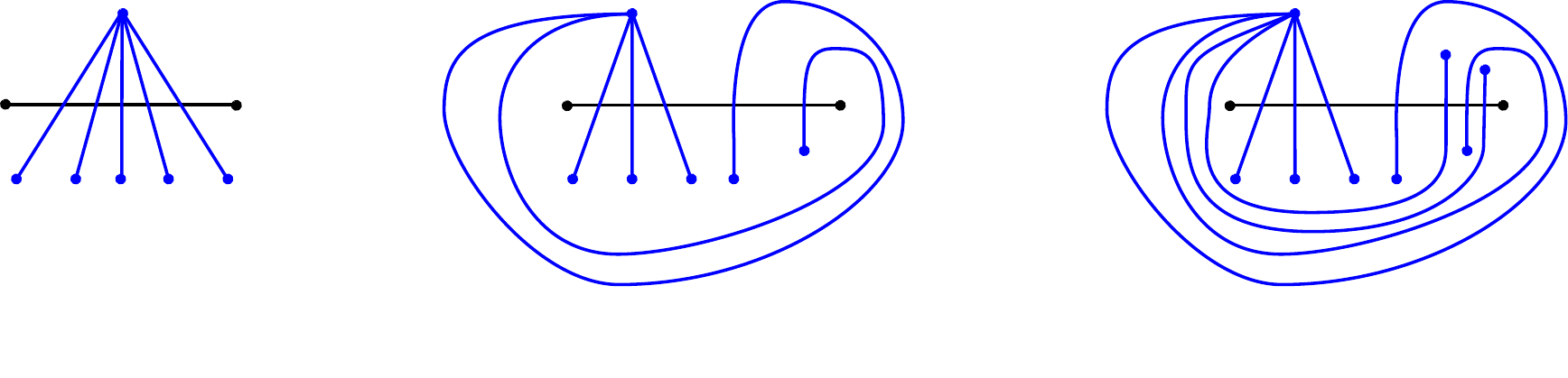}}%
  \end{picture}%
\endgroup%

    \caption{(a) A strongly fan-planar drawing.  (b) A weakly fan-planar drawing that is not strongly fan-planar.  (c) A fan-crossing drawing that is not weakly fan-planar.}
    \label{fig:fan}
\end{figure}

In a very recent manuscript, Hendrey, Karol, and Wood \cite{DBLP:journals/corr/abs-2507-22395} define \emph{$k$-cover-planar} drawings as normal drawings in which, for every edge $e$, the set of edges intersecting $e$ can be covered by at most $k$ fans.  Similarly, they define \emph{$k$-matching-planar} drawings as normal drawings in which, for every edge $e$, the set of edges intersecting~$e$ do not contain a matching of size~$k+1$.
Obviously, a fan-crossing drawing is $1$-cover-planar, a $k$-planar drawing is $k$-cover-planar, a $k$-cover-planar drawing is $k$-matching-planar, and a $k$-matching-planar drawing is $(2k)$-cover-planar.
Let the \emph{$k$-cover-planar} crossing number be the minimum number of crossings over all $k$-cover-planar drawings, and similarly let the \emph{$k$-matching-planar} crossing number be the minimum number of crossings over all $k$-matching-planar drawings.

Also, a normal drawing is \emph{fan-crossing free} if no two edges crossing the same edge share a vertex, and the \emph{fan-crossing free crossing number} is the minimum number of crossings over all fan-crossing free drawings.

\FPT\ algorithms, all using Courcelle's theorem, exist for the fan-crossing and fan-crossing free crossing numbers, by M\"unch and Rutter~\cite{patterncrossing24}, and for the weakly fan-planar crossing number, by Hamm, Klute, and Parada~\cite{DBLP:journals/corr/abs-2412-13092}, and it seems that with further adjustments the latter authors could also handle the strongly fan-planar version.  
We solve all these variants, again with a better runtime.  For the strongly fan-planar version, we use the fact that our framework allows for surfaces with boundary, and thus can easily capture the ``infinite face'' of the drawing; indeed, strong fan-planarity is invariant under self-homeomorphisms of the plane, but not of the sphere (e.g., the drawings in \Cref{fig:fan}(a) and (b) are homeomorphic when seen on the sphere).

\begin{proposition}\label{prop:fan}
Let $k\geq1$ and $t\geq0$ be integers and let $G$ be a graph of size~$n$.  Deciding whether $G$ has strongly fan-planar, weakly fan-planar, fan-crossing, $k$-matching-planar, $k$-cover-planar, or fan-crossing free crossing number at most~$t$ in the plane can be done in time $2^{\poly(t)}\cdot n^2$, or in time $f(t)\cdot n$ for some function~$f$.
\end{proposition}
\begin{proof}
  The proof is very similar to that of \Cref{prop:kplanar-etc}, specialized to the plane.
  We need to be able to recognize strongly fan-planar, weakly fan-planar, fan-crossing, or fan-crossing free drawings in polynomial time, which is not difficult from their definitions.
  To recognize $k$-matching-planar drawings, we employ a standard maximum matching algorithm for every edge; as before, it suffices to do this in polynomial time.  The case of recognizing $k$-cover-planar drawings is more interesting since no fully polynomial recognition algorithm is known; the problem is equivalent to solving $k$-vertex-cover, which can easily be solved in time $\ca O(2^kt)$ for each of the at most $2t$ crossed edges since no more than $t$ edges cross a single edge. However, this runtime does not impair the overall runtime of \Cref{thm:main}.
  \end{proof}
Extensions of \Cref{prop:fan} to arbitrary surfaces are obviously possible.

\subparagraph{Drawings that may not be normal.}
To illustrate the applicability of our framework to non-normal drawings, we introduce the following example. For an integer $k\geq2$, we define the \emph{$k$-intersecting crossing number} of a graph $G$ to be the minimum value~$t$ such that there is a drawing of $G$ with at most~$t$ intersection points in which the multiplicity of every intersection point is at most~$k$.  For every graph $G$, the traditional crossing number of $G$ is at most $k\choose2$ times its $k$-intersecting crossing number, and this bound can be tight.  We prove:

\begin{proposition}\label{prop:k-intersecting}
Let $\surf$ be a surface with topological size~$s$, let $k\geq2$ and $t\geq0$ be integers, and let $G$ be a graph of size~$n$.  Deciding whether $G$ has $k$-intersecting crossing number at most~$t$ in $\surf$ can be done in time $2^{\poly(s+k+t)}\cdot n^2$, or in time $f(s+k+t)\cdot n$ for some function~$f$.
\end{proposition}
\begin{proof}
  The proof is again very similar to the proof of \Cref{prop:kplanar-etc}.  Let $\DD$ be the class of (possibly non-normal) drawings on~$\surf$  with at most $t$ intersection points, all of multiplicity at most~$k$; $\DD$ is clearly stable, and testing membership in~$\DD$ can be done in polynomial time.  Moreover, $r:=kt$ is an upper bound on the intersecting size of~$\DD$.  Hence, our conclusion follows again by applying \Cref{thm:main}.
\end{proof}

Note that our definition allows for tangential intersections of edges.  However, the same argument would work in the variant where we forbid tangential intersections.

We refer to \Cref{sub:nontraditional} for more examples of different methods for counting crossings.

\subsection{Crossing problems with colored edges: Joint crossing number and generalizations}\label{sub:coloredges}

We now turn our attention to problems in which the colors of intersecting edges play a role. As a natural example, we choose the \emph{joint crossing number} problem introduced a while ago by Negami~\cite{DBLP:journals/jgt/Negami01}.  In this problem, the goal is to embed two input graphs simultaneously on the same surface while minimizing the number of crossings between them.  It is \NP-hard for any fixed Euler genus $g\geq6$~\cite{DBLP:conf/isaac/HlinenyS15,DBLP:conf/mfcs/Hlineny25}, and the complexity parameterized by the solution value has not been studied prior to our work.  We actually solve the following new problem generalizing the joint crossing number:
\medskip
\Prob{\cccn{}}%
{A surface~$\surf$; a symmetric matrix~$M$ of size $c\times c$ with nonnegative integer values; a colored graph~$G$ with colors in~$\{1,\ldots,c\}$}%
{Is there a normal colored drawing $D$ of~$G$ in~$\surf$ such that, for each pair $i,j\!\in\!\{1,\ldots,c\}$, the number of crossings involving two edges, one colored $i$ and the other~$j$, is at most~$M_{i,j}$?}

\begin{proposition}\label{prop:CCCN}
The \cccn{} can be solved in time $2^{\poly(s+t)}\cdot n^2$, where $s$ is the topological size of~$\surf$, $t=\sum_{1\leq i\leq j\leq c}M_{i,j}$, and $n$ is the size of~$G$.
\end{proposition}
\begin{proof}
  Without loss of generality, we assume that $c\le t+1$, since merging the colors corresponding to empty rows of~$M$ will not change the problem.  Let $\DD$ be the class of normal drawings on~$\surf$ with at most $M_{i,j}$ crossings involving an edge of color~$i$ and an edge of color~$j$.  Clearly $\DD$ is stable, and one can decide whether a drawing is in~$\DD$ in time polynomial in the representation of the drawing.  Moreover, $r:=2t$ is an upper bound on the intersecting size of~$\DD$.  \Cref{thm:main} yields the result.
\end{proof}

\begin{corollary}\label{cor:jointcr}
  One can decide whether the joint crossing number of graphs $G_1$ and $G_2$, each embeddable in a surface~$\surf$, is at most~$t$ in time $2^{\poly(s+t)}\cdot n_0^2$ or in time $f(s+t)\cdot n_0$ for some function~$f$, where $s$ is the topological size of~$\surf$ and $n_0$ is the size of~$G_1\cup G_2$.
\end{corollary}
\begin{proof}
  We apply \Cref{prop:CCCN} as follows.  For $i=1,2$, we use color~$i$ for~$G_i$, define $G$ to be the disjoint union of $G_1$ and~$G_2$, let $M=\begin{pmatrix}0&t\\t&0\end{pmatrix}$.
\end{proof}

Many variants of the above problem exist in which one also prescribes the rotation system of the input graph.  We defer the discussion on how we handle such problems to \Cref{sec:metasection-rot}.

\subsection{Non-traditional methods of counting crossings}\label{sub:nontraditional}

We now look at different problems that do not count the crossings simply ``one by one''.

\subparagraph{Edge crossing number.}  Perhaps the most natural problem that fits into this category is the edge crossing number. The \emph{edge crossing number} of a graph $G$ is the smallest integer~$t$ such that $G$ has a normal drawing having $t$ edges with at least one crossing. 
Computing the edge crossing number is \NP-hard~\cite{DBLP:conf/gd/BalkoHMOVW24}, and according to Schaefer~\cite[Section 3.2: Edge crossing number]{DBLP:journals/combinatorics/SchaeferDS21}, no \FPT\ algorithm parameterized by the solution size $t$ is known even in the plane.

\begin{proposition}\label{prop:edgecrossn}%
  Let $\surf$ be a surface with topological size~$s$, let $t\geq0$ be an integer, and let $G$ be a graph of size~$n$.  Deciding whether $G$ has edge crossing number at most~$t$ in~$\surf$ can be done in time $2^{\poly(s+t)}\cdot n^2$, or in time $f(s+t)\cdot n$ for some function~$f$.
\end{proposition}
\begin{proof}
  It is well known, see, e.g., Schaefer~\cite{schaefer2017crossing}, that any drawing~$D$ of a graph~$G$ with at most $t$ edges carrying an intersection point can be turned into a normal drawing~$D$ in which no two edges cross twice, without introducing new crossed edges.  After removing the possible self-crossings of edges, we obtain a drawing of~$G$ with at most $t$ edges carrying an intersection point and, moreover, with at most $t\choose2$ intersection points.
  
  Let $\DD$ be the class of normal drawings on~$\surf$ such that at most $t$ edges carry an intersection point, and with at most $t\choose2$ intersection points.  By the previous paragraph, $G$ has edge crossing number at most~$t$ if and only if $\DD$ contains a drawing of~$G$.  Moreover, $\DD$ is stable, and membership in~$\DD$ can be tested in polynomial time.  Also, $r:=2{t\choose2}$ is an upper bound on the intersecting size of~$\DD$.  \Cref{thm:main} concludes.
\end{proof}
A natural edge-colored generalization can be solved by a straightforward combination of the ideas in the proof of \Cref{prop:CCCN} and of the previous proof.

\subparagraph{Odd and pair crossing numbers.}

The \emph{pair (resp.\ odd) crossing number} of a graph $G$ is the minimum $t$ such that there exists a normal drawing of $G$ in which at most $t$ pairs of edges mutually cross (resp.\ cross an odd number of times).
While the question whether the pair crossing number coincides with the traditional crossing number is one of the biggest open problems in crossing numbers, Pelsmajer, Schaefer, and \v{S}tefankovi\v{c}~\cite{DBLP:journals/dcg/PelsmajerSS08} proved that the odd crossing number can be lower than the traditional crossing number.

\FPT\ algorithms for the pair and odd crossing numbers were given again by Pelsmajer, Schaefer, and \v{S}tefankovi\v{c}~\cite{DBLP:conf/gd/PelsmajerSS07a}, with an unspecified dependency on~$t$; however, because these algorithms are based on an adaptation of Grohe's algorithm \cite{DBLP:journals/jcss/Grohe04}, they are quadratic in the size~$n$ of the input and the dependency on~$t$ is at least an exponential tower of height four.  We get faster algorithms:

\begin{proposition}\label{prop:pair-odd-cr}
  Let $\Pi\in\{$`pair', `odd'$\}$, let $t\geq0$ be an integer, and let $G$ be a graph of size~$n$. Deciding whether $G$ has $\Pi$-crossing number at most~$t$ in the \emph{plane} can be done in time $2^{2^{\OO(t)}}\!\!\cdot n^2$, or in time $f(t)\cdot n$ for some function~$f$.
\end{proposition}
\begin{proof}
  The approach is analogous to the proof of \Cref{prop:edgecrossn}, with the additional ingredient given by the bounds by Pelsmajer et al.~\cite{DBLP:conf/gd/PelsmajerSS07a}: in the case of the plane, if $G$ is of pair (respectively, odd) crossing number~$t$, then there exists a normal drawing $D$ of $G$ achieving this optimum such that the number of crossings in $D$ is at most $t2^t$ (respectively,~$9^t$).

  Let $\DD$ be the class of normal drawings in the plane (or sphere) with $\Pi$-crossing number at most~$t$ and with at most $\max(t2^t,9^t)$ intersection points.   By the previous paragraph, $G$ has $\Pi$-crossing number at most~$t$ if and only if $\DD$ contains a drawing of~$G$.  Moreover, $\DD$ is stable, and membership in~$\DD$ can be tested in polynomial time.  Also, $r:=2\max(t2^t,9^t)$ is an upper bound on the intersecting size of~$\DD$.  \Cref{thm:main} concludes.
\end{proof}

We remark that the same approach would extend to surfaces, assuming some bounds analogous to those given by Pelsmajer et al.~\cite{DBLP:conf/gd/PelsmajerSS07a} for the plane.

\subsection{Allowing edge removals and vertex splits before drawing}\label{sub:delete-split}

Finally, we consider crossing number variants that allow graph simplifications before drawing.

\subparagraph{Skewness.}

The \emph{skewness} of a graph $G$ is the smallest number of edges whose removal from~$G$ leaves a planar graph. Deciding whether the skewness of $G$ is at most $q$ is NP-complete~\cite{liu-deletion}, and linear-time FPT algorithms are claimed, without details, by Kawarabayashi and Reed~\cite{DBLP:conf/stoc/KawarabayashiR07} and Jansen, Lokshtanov, and Saurabh~\cite{jls-nopa-14}.  We generalize the problem as follows.

\Prob{Color-Constrained Skewness with Crossings}%
{A surface~$\surf$; non-negative integers $t$ and $q_1,\ldots,q_c$; a colored graph~$G$ with colors in~$\{1,\ldots,c\}$.}%
{Can we remove, for each $i=1,\ldots,c$, at most $q_i$ edges of color~$i$ from $G$, such that the resulting graph has crossing number at most $t$ in~$\surf$?}

\begin{proposition}\label{prop:skewness}
  The {\sc Color-Constrained Skewness with Crossings} problem can be solved in time $2^{\poly(s+t+q)}\cdot n^2$ or in time $f(s+t+q)\cdot n$ for some function~$f$, where $s$ is the topological size of~$\surf$, $q=\sum_iq_i$, and $n$ is the size of~$G$.
\end{proposition}
\begin{proof}
  Without loss of generality, we assume that $c\leq q+1$, since the colors $i$ with $q_i=0$ can be merged into one.  Let $\DD$ be the class of normal drawings on~$\surf$ with at most $t$ crossings.  Then $\DD$ is a stable class, membership in~$\DD$ can be tested in polynomial time, and $r:=2t$ is an upper bound on the intersecting size of~$\DD$.  There remains to apply \Cref{thm:main} with $p=0$ and with the specified values of $q_1,\ldots,q_c$.
\end{proof}

\subparagraph{Splitting number.}  The smallest integer $p$ such that a graph obtained from the given graph~$G$ by $p$ successive vertex splits is embeddable in~$\surf$ is called the \emph{splitting number} of~$G$ in~$\surf$ \cite{splitsurf,hjr-sncg-85}. This problem is \NP-hard~\cite{DBLP:journals/dam/FariaFN01}, already in the plane.
N{\"{o}}llenburg, Sorge, Terziadis, Villedieu, Wu, and Wulms~\cite{DBLP:conf/gd/NollenburgSTVWW22} proved that the property of having splitting number at most~$p$ is minor-monotone in any fixed surface, and so the {$\surf$-splitting number} has a nonuniform \FPT\ algorithm parameterized by $p$ using the theory of graph minors.  We generalize and improve the latter result to a uniform \FPT\ algorithm:
\begin{proposition}\label{prop:splitting}
  Let $\surf$ be a surface with topological size~$s$, let $t,p\geq0$ be integers, and let $G$ be a graph of size~$n$.  Deciding whether $G$ has, after at most $p$ vertex splits, crossing number at most $t$ in~$\surf$ can be done in time $2^{\poly(s+t+p)}\cdot n^2$, or in time $f(s+t+p)\cdot n$ for some function~$f$.
\end{proposition}
\begin{proof}
  We use the same definitions of~$\DD$ and~$r$ as in the proof of \Cref{prop:skewness}, and apply \Cref{thm:main} with $q_1=\ldots=q_c=0$ and the specified value of~$p$.
\end{proof}

\section{Fixing the rotation system: framework and applications}\label{sec:metasection-rot}

A natural restriction, when minimizing the number of crossings of an input graph, is to prescribe the clockwise cyclic order of edges around each vertex in the drawing---the \emph{rotation system}.  It turns out that \Cref{thm:main} also allows us to capture such constraints, at the cost of a mild additional condition on the class~$\DD$.  In order to make our framework more directly applicable, we first show a quite general consequence of \Cref{thm:main}, which can then be instantiated to obtain FPT algorithms for several well-studied problems with fixed rotation system.  In this section, we only consider orientable surfaces.

\subsection{Statement of the general result}

In our general result, we consider an input graph~$G$ together with a prescribed rotation system, and we consider drawings of~$G$ in which the input rotation system may be flipped independently for each connected component of~$G$.  We first need some terminology.

Given a graph~$G$, a \emph{rotation system} is a family of cyclic permutations $\pi=\big(\pi_v:v\in V(G)\big)$, one for each vertex of~$G$, such that the domain of the permutation~$\pi_v$ is the set of edges incident to~$v$ in~$G$.  (If a loop is incident to~$v$, it should appear twice in that cyclic permutation---or more formally, we should consider a cyclic permutation of the \emph{half-edges} of~$G$ incident to~$v$.)  Let $\pi_v^{-1}$ be the inverse of~$\pi_v$; also, let $\pi^{-1}:=\big(\pi^{-1}_v:v\in V(G)\big)$.  Now, let $D$ be a drawing of~$G$ in an oriented surface~$\surf$.  The \emph{rotation system} of~$D$ is obtained by recording the clockwise ordering of the edges around each vertex in~$D$.  Let $G_0$ be a connected component of~$G$.  We say that the rotation system of~$G_0$ in~$D$ is \emph{induced} by $\pi$ if, for each vertex~$v$ of~$G_0$, the cyclic ordering of the edges around~$v$ in~$D$ equals~$\pi_v$.

In this section, we study the following variant of the \DCP:
\Prob{\DCROT{}}%
{a positive integer $c$; a colored graph~$G$ with colors in~$\{1,\ldots,c\}$; a rotation system of~$G$}%
{Is there a colored drawing~$D\in\DD$ of the colored graph~$G$ such that, for each connected component~$G_0$ of~$G$, the rotation system of~$G_0$ in~$D$ is induced either by $\pi$ or by~$\pi^{-1}$?}

In order to obtain an FPT algorithm for this problem, we need to strengthen the stability condition for~$\DD$.  Let $D$ be a drawing of a graph~$G$.  Let $G'$ be a matching with the same edge set as~$G$.  We define a drawing~$D'$ of~$G'$ as follows.  In~$D$, each edge~$e$ of~$G$ is drawn as a curve~$c$, and in~$D'$ the corresponding edge of~$G'$ is drawn as a subcurve~$c'$ of~$c$ obtained by removing from~$c$ small initial and final parts of~$c$; in particular, $c'$ contains all intersection points of~$c$.  The drawing~$D'$ is unique up to homeomorphism, and we call it the \emph{detachment} of~$D$.  A class $\DD$ of drawings is \emph{detachably stable} if (1) $\DD$ is stable and (2) a drawing~$D$ lies in~$\DD$ \emph{if and only if} its detachment lies in~$\DD$.  

We obtain the following result:
\begin{theorem}\label{thm:mainrot}
  Let $c$ be a positive integer and $\surf$ be an orientable surface of topological size~$s$.   Let $\DD$ be a detachably stable class of drawings with colors in~$\{1,\ldots,c\}$ on~$\surf$, specified by an integer~$r$ that is an upper bound on the intersecting size of~$\DD$, and by an algorithm that, in time $\OO(\delta(j))$, decides membership in~$\DD$ of a representation of a colored drawing of size~at~most~$j$.

  Then, the \DCROT{} can be solved in time $2^{\poly(c+s+r)}\cdot\delta(\OO(cr^2s^2))\cdot n^2$, or in $f(c+s+r)\cdot n$ time for some function~$f$, where $n$ is the size of~$G$.
\end{theorem}  

The proof of \Cref{thm:mainrot} is a reduction from the \DCP{}, and in particular the remarks listed after \Cref{thm:main} are also valid for \Cref{thm:mainrot}.

\subsection{Proof of \Cref{thm:mainrot}}

We now turn to the proof of \Cref{thm:mainrot}.  Given the input $(G,\pi)$ to the \DCP{}, we thicken~$G$ using grids in order to make it rigid according to the input rotation system~$\pi$.

We start with some preliminaries.  Let $m,n\ge2$ be two integers.  The \emph{grid with $m$ rows and $n$ columns} is the Cartesian product of two paths on $m$ and $n$ vertices. That is, its vertex set is $\{(i,j): 1\leq i\leq m,\, 1\leq j\leq n\}$, and vertices $(i,j)$ and $(i',j')$ are adjacent if and only if $|i-i'|+|j-j'|=1$.  A \emph{cell} is a \emph{directed} cycle of the grid of the form $(i,j),(i+1,j),(i+1,j+1),(i,j+1)$.  The $k$-th \emph{vertical path} is induced by the vertices $(i,k)$, $i=1,\ldots,m$, while the $k$-th \emph{horizontal path} is induced by the vertices $(k,j)$, $j=1,\ldots,n$.  The \emph{top path} of the grid is the first horizontal path.  A \emph{horizontal (respectively, vertical) strip} is induced by the vertices of two consecutive horizontal (respectively, vertical) paths.  A cell of a grid embedded in~$\surf$ is \emph{flat} if it bounds a face of the grid that is homeomorphic to a disk in~$\surf$.  A strip is \emph{flat} if all its cells are.  The following lemma is a simple consequence of a lemma by Geelen, Richter, and Salazar~\cite{DBLP:journals/ejc/GeelenRS04}.
\begin{lemma}\label{lem:GRSgrids}
   Let $G$ be the ($m\times n$)-grid, where $m,n\geq3$, embedded in a (possibly non-orientable) surface $\surf$ with topological size~$s$.  Then, the number of cells of $G$ that are not flat is at most $9s$.
\end{lemma}
\begin{proof}
  Geelen, Richter, and Salazar {\cite[Lemma~1]{DBLP:journals/ejc/GeelenRS04}} proved that, under the hypotheses of the lemma but assuming, in addition, that $\surf$ has no boundary, the number of noncontractible cells in~$G$ is at most $9g$, where $g$ is the Euler genus of~$\surf$.  (Actually, they only consider square grids, but it is immediate from their proof that the condition $m,n\geq3$ is sufficient for the result to hold.)

  Since $G$ is connected, we may without loss of generality assume $\surf$ connected.  Let $g$ and~$b$ be the genus and number of boundary components of~$\surf$, respectively, so that $s=1+g+b$.  Let $\surf'$ be obtained from~$\surf$ by attaching a crosscap (a Möbius strip) to each boundary component of~$\surf$; thus $\surf'$ has no boundary, and its Euler genus is~$s-1$.  By the above lemma by Geelen, Richter, and Salazar, at most $9(s-1)$ cells of~$G$ are noncontractible in $\surf'$.  Each remaining cell~$C$ is contractible in~$\surf'$, hence separating in~$\surf'$, hence bounds a face of~$G$ in~$\surf'$, because $G-C$ is connected.

  So, each cell of~$G$, with at most $9(s-1)$ exceptions, is contractible and bounds a face in~$\surf$.  If every such cell bounds a face that is homeomorphic to a disk, then we are done.  Otherwise, there is a cell~$C$ that splits $\surf$ into a disk and another surface, and all the remaining part of~$G$ lies inside the disk.  Since $G$ is 3-connected, it has a unique combinatorial embedding in the sphere, so all cells of~$G$ except~$C$ bound a face homeomorphic to a disk in~$\surf$.  The result follows since $s\ge1$.
\end{proof}

  Let $(G,\pi)$ be an instance of the \DCROT{}.  We define a class $\DD''$ and a graph $G''$ such that $(G,\pi)$ is a positive instance of the \DCROT{} if and only if $G''$ is a positive instance of the \DxC{$\DD''$} \textsc{Problem}.

\begin{figure}%
    \centering
    \def\svgwidth{\linewidth}
\begingroup%
  \makeatletter%
  \providecommand\color[2][]{%
    \errmessage{(Inkscape) Color is used for the text in Inkscape, but the package 'color.sty' is not loaded}%
    \renewcommand\color[2][]{}%
  }%
  \providecommand\transparent[1]{%
    \errmessage{(Inkscape) Transparency is used (non-zero) for the text in Inkscape, but the package 'transparent.sty' is not loaded}%
    \renewcommand\transparent[1]{}%
  }%
  \providecommand\rotatebox[2]{#2}%
  \newcommand*\fsize{\dimexpr\f@size pt\relax}%
  \newcommand*\lineheight[1]{\fontsize{\fsize}{#1\fsize}\selectfont}%
  \ifx\svgwidth\undefined%
    \setlength{\unitlength}{645.13413462bp}%
    \ifx\svgscale\undefined%
      \relax%
    \else%
      \setlength{\unitlength}{\unitlength * \real{\svgscale}}%
    \fi%
  \else%
    \setlength{\unitlength}{\svgwidth}%
  \fi%
  \global\let\svgwidth\undefined%
  \global\let\svgscale\undefined%
  \makeatother%
  \begin{picture}(1,0.27501964)%
    \lineheight{1}%
    \setlength\tabcolsep{0pt}%
    \put(0,0){\includegraphics[width=\unitlength,page=1]{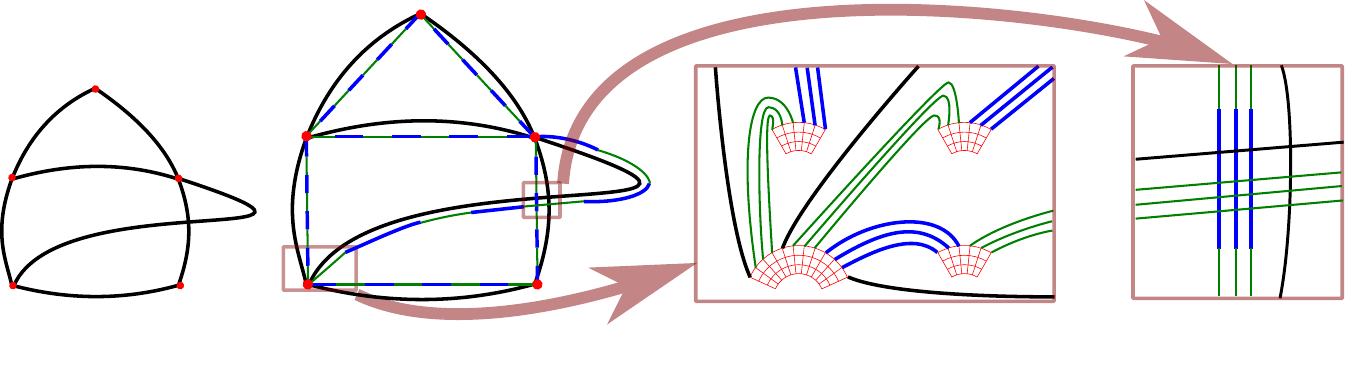}}%
    \put(0.30056143,0.00408709){\color[rgb]{0,0,0}\makebox(0,0)[lt]{\lineheight{1.49999988}\smash{\begin{tabular}[t]{l}(b)\end{tabular}}}}%
    \put(0.61449311,0.00408709){\color[rgb]{0,0,0}\makebox(0,0)[lt]{\lineheight{1.49999988}\smash{\begin{tabular}[t]{l}(c)\end{tabular}}}}%
    \put(0.89390136,0.00408709){\color[rgb]{0,0,0}\makebox(0,0)[lt]{\lineheight{1.49999988}\smash{\begin{tabular}[t]{l}(d)\end{tabular}}}}%
    \put(0.04989844,0.00408709){\color[rgb]{0,0,0}\makebox(0,0)[lt]{\lineheight{1.49999988}\smash{\begin{tabular}[t]{l}(a)\end{tabular}}}}%
  \end{picture}%
\endgroup%

    \caption{The construction of the graphs $G'$ and~$G''$ for the proof of \Cref{prop:CCCNrot}, illustrated with $r=4$ and $k=3$. (a) The initial graph~$G$, drawn on~$\surf$ with a single intersection point, a crossing.  In this example, all edges have a single color, black.  (b) The construction of~$G'$ from~$G$.  (c,d) The construction of~$G''$ from~$G'$, shown only in the bottom left part of~$G'$ and in the neighborhood of the unique crossing between two edges of~$G$.}
    \label{fig:CCCNrot}
\end{figure}

  The construction of~$G''$ is illustrated in \Cref{fig:CCCNrot}.  This graph is essentially a ``thickened version'' of~$G$.  It has three more colors than~$G$; color~$c+1$ is called \emph{red}, $c+2$ \emph{green}, and $c+3$ \emph{blue}.  Each vertex of~$G$ corresponds to a red grid in~$G''$; each edge of~$G$ corresponds to an edge, of the same color, in~$G''$ and additionally to a set of red, green, and blue edges connecting the corresponding red grids, to force the consistency of orientation between the two corresponding red grids.  As suggested by \Cref{fig:CCCNrot}, a drawing~$D$ of~$G$ in which the rotation system of each connected component is induced either by $\pi$ or~$\pi^{-1}$ leads naturally to a drawing~$D''$ of~$G''$ whose crossings arise from the crossings of~$D$.  The class~$\DD''$ is defined, also with the additional colors, in such a way that it contains the drawing~$D''$, for each $D\in\ca D$.  In particular, the red edges cannot bear any intersection point, and no two green or two blue edges intersect.

  We now provide the details of the construction of~$G''$.  Recall that $r$ is an upper bound on the intersecting size of $\DD$.  Let $k:=9s+3$.
  \begin{itemize}
    \item In a first step (\Cref{fig:CCCNrot}(b)), we compute a colored graph~$G'$ with a rotation system~$\pi'$.  The graph~$G'$ is obtained from~$G$ by adding, for every edge~$e$ of~$G$, a path~$P'_e$ of length~$2r$ with the same endpoints as~$e$, and coloring its edges alternatively green and blue; edge $e$ keeps its original color, and is called an \emph{original edge}.  We extend the rotation system~$\pi$ of~$G$ in order to provide $G'$ with a rotation system~$\pi'$ as follows.  There is a unique possible cyclic order of edges for the new vertices of~$G'-G$, which have degree two.   So, for each original edge $e=uv$, there remains to insert the first and last edges of each path~$P'_e$ in $\pi_u$ and~$\pi_v$, respectively.  We insert the first edge of~$P'_e$ into the cyclic order $\pi_u$ just \emph{before} the occurrence of~$e$ in~$\pi_u$, and we insert the last edge of~$P'_e$ into the cyclic order $\pi_v$ just \emph{after} the occurrence of~$e$ in~$\pi_v$.  Intuitively, this means that $e$ and~$P'_e$ run parallel to each other ($P'_e$ to the left of~$e$; the construction would thus change if the endpoints $u$ and~$v$ of~$e$ were swapped, but this is insignificant).  Let $\pi'$ be the resulting rotation system.
    \item
    The remaining part of the construction ``thickens'' $G'$ in a way that we now describe (\Cref{fig:CCCNrot}(c)).  For each vertex $v$ of~$G'$ we do the following.  We create a copy~$G''_v$ of the $(k\times n_v)$-grid, for some integer~$n_v$ specified below, in which all edges are colored red.  Let $e_1,\ldots,e_d$ be the edges of~$G'$ incident to~$v$ in the prescribed clockwise order~$\pi'_v$, and $T''_v$ the top path of the grid~$G''_v$.  We partition the vertex set of~$T''_v$ into sets of consecutive vertices $T''_{v,e_1},\ldots,T''_{v,e_d}$, in this order along~$T''_v$, such that $T''_{v,e_i}$ has size one if $e_i$ is an original edge, and size $k$ otherwise.  In particular, $n_v=|T''_v|$ is the sum of these sizes.  (We may clearly assume from the beginning that $G$ has no isolated vertex, so $n_v\ge3$.)
    \item For each original edge $e=uv\in E(G')$, we connect $T''_{u,e}$ and~$T''_{v,e}$ via an edge of the same color as~$e$, still called an original edge.  For each non-original edge $e=uv\in E(G')$, we create a set $Q''_e$ of~$k$ edges, of the same color as~$e$ (either green or blue), connecting, for $i=1,\ldots,k$, the $i$-th vertex of~$T''_{u,e}$ with the $(k+1-i)$-th vertex of~$T''_{v,e}$.
\end{itemize}
  Let $G''$ be the resulting graph.

  There remains to define the class~$\DD''$.  Recall that a \emph{crossing} is an intersection point of multiplicity two that cannot be removed by a local perturbation.  A drawing~$D''$ with colors in~$\{1,\ldots,c+3\}$ belongs to~$\DD''$ if and only if all the following conditions are satisfied:
  \begin{itemize}
      \item removing all red, green, and blue edges from~$D''$ yields a drawing in~$\DD$;
      \item each red edge in~$D''$ bears no intersection point;
      \item each green edge in~$D''$ bears only intersection points that are crossings, and not with a red or green edge; 
      symmetrically, each blue edge in~$D''$ bears only intersection points that are crossings, and not with a red or blue edge; 
      \item the number of crossings involving a green or a blue edge (or both) is bounded from above by~$(k+1)^2r^2$, where $r$ is an upper bound on the intersecting size of~$\DD$.
  \end{itemize}
  
  In the two following lemmas, we prove that the reduction is valid, namely, that $(G,\pi)$ is a positive instance of the \DCROT{} if and only if $G''$ is a positive instance of the \DxC{$\DD''$} \textsc{Problem}.

  \begin{lemma}\label{lem:rot1}
  Let $D$ be a colored drawing of~$G$ that belongs to~$\DD$.  Assume that, for each connected component~$G_0$ of~$G$, the rotation system of~$G_0$ in~$D$ is induced by~$\pi$ or~$\pi^{-1}$.  Then some drawing of~$G''$ belongs to~$\DD''$.
  \end{lemma}
  \begin{proof}
  Each connected component of~$G$ is either \emph{unflipped}, if in the drawing~$D$ its rotation system is induced by~$\pi$, or \emph{flipped}, if it is induced by~$\pi^{-1}$.  First, we define a drawing~$D'$ of~$G'$ as follows (\Cref{fig:CCCNrot}(a,b)).  We draw each path $P'_e$ very close to the corresponding original edge~$e$, on the appropriate side of~$e$, so that the rotation system of~$D'$ is precisely~$\pi'$ on unflipped connected components and $\pi^{\prime -1}$ on flipped connected components.  Moreover, we do this in such a way that all the intersection points on~$P'_e$ are crossings, and that each intersection point of multiplicity~$\ell$ in~$D$ corresponds to at most $\ell$ crossings of~$P'_e$ with original edges and $\ell-1$ crossings of~$P'_e$ with some paths~$P'_f$.  See \Cref{fig:highmult} for an illustration, and its caption for the details of the construction.
\begin{figure}%
    \centering
    \def\svgwidth{.4\linewidth}
\begingroup%
  \makeatletter%
  \providecommand\color[2][]{%
    \errmessage{(Inkscape) Color is used for the text in Inkscape, but the package 'color.sty' is not loaded}%
    \renewcommand\color[2][]{}%
  }%
  \providecommand\transparent[1]{%
    \errmessage{(Inkscape) Transparency is used (non-zero) for the text in Inkscape, but the package 'transparent.sty' is not loaded}%
    \renewcommand\transparent[1]{}%
  }%
  \providecommand\rotatebox[2]{#2}%
  \newcommand*\fsize{\dimexpr\f@size pt\relax}%
  \newcommand*\lineheight[1]{\fontsize{\fsize}{#1\fsize}\selectfont}%
  \ifx\svgwidth\undefined%
    \setlength{\unitlength}{121.10003525bp}%
    \ifx\svgscale\undefined%
      \relax%
    \else%
      \setlength{\unitlength}{\unitlength * \real{\svgscale}}%
    \fi%
  \else%
    \setlength{\unitlength}{\svgwidth}%
  \fi%
  \global\let\svgwidth\undefined%
  \global\let\svgscale\undefined%
  \makeatother%
  \begin{picture}(1,1.00002224)%
    \lineheight{1}%
    \setlength\tabcolsep{0pt}%
    \put(0,0){\includegraphics[width=\unitlength,page=1]{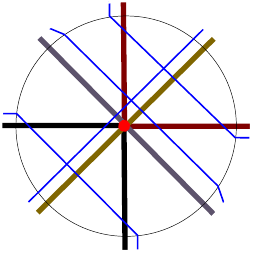}}%
    \put(0.27671108,0.89680486){\color[rgb]{0,0,0}\makebox(0,0)[lt]{\lineheight{1.5}\smash{\begin{tabular}[t]{l}$d$\end{tabular}}}}%
  \end{picture}%
\endgroup%

    \caption{An illustration of the construction of the drawing~$D'$ of~$G'$ in the proof of \Cref{lem:rot1}, in the case where $\ell=4$ and all paths~$P'_e$ are drawn in thin blue lines.  Consider a small disk~$d$ centered at a point~$x$ of multiplicity~$\ell\ge2$ in~$D$.  Up to homeomorphism, we assume that $d$ is a Euclidean disk. We may assume that inside~$d$, the image of~$G$ is the union of radii, whose intersections with the boundary of~$d$ are the vertices of a regular $2\ell$-gon, so that all edges of~$G$ meet at~$x$.  Each path~$P'_e$ enters~$d$ as many times as edge~$e$ goes through~$x$, and we make sure that each time it does so, it uses a straight line segment.  If the intersection points of these paths with the boundary of~$d$ are in general position, then all intersection points except~$x$ are crossings, and each line segment crosses at most half of the~$2\ell$ radii and at most once each of the $\ell-1$ other segments.}
    \label{fig:highmult}
\end{figure}

  There remains to specify where to place the inner vertices of each path~$P'_e$ along~$P'_e$.  Since the sum of the multiplicities of the crossings of~$G$ is at most~$r$, when walking along a path~$P'_e$, one visits at most $r$ crossings with the paths~$P'_f$.  Recall that $P'_e$ is a path made of $2r$~edges.  We put the $i$-th crossing point along~$P'_e$ either on the $(2i-1)$-th or $(2i)$-th edge of~$P'_e$ (one of which is green and the other is blue), in such a way that each crossing between two paths, $P'_e$ and~$P'_f$, is a crossing between a green edge and a blue edge. See \Cref{fig:CCCNrot}(b).  (We do not worry about the crossings with the original edges.)

  Now, from $D'$, we compute a drawing~$D''$ of~$G''$ by replacing each vertex~$v$ of~$G''$ with the $(k\times n_v)$-grid $G''_v$ in a small neighborhood of the image of~$v$, with its orientation reversed if and only if the corresponding connected component is flipped.  Similarly, for each edge $e$ of~$G'$, we build the corresponding edge(s) of~$G''$ in a small neighborhood of the image of~$e$.  See \Cref{fig:CCCNrot}(c,d).

  To conclude the proof, it suffices to check that $D''$ belongs to~$\DD''$:
  \begin{itemize}
  \item Let $D_1$ denote the drawing obtained from $D''$ by removing all red, green, and blue edges. Then $D_1$ is the detachment of~$D\in\DD$, and hence lies in~$\DD$ since $\DD$ is detachably stable.
  \item Each red edge in~$D''$ bears no intersection point.
  \item Each green edge in~$D''$ bears only intersection points that are crossings, and not with a red or green edge.  Symmetrically, each blue edge in~$D''$ bears only intersection points that are crossings, and not with a red or blue edge.
  \item We now bound the number of crossings involving a green or blue edge (or both).  Each intersection point of multiplicity~$\ell$ in~$D$ contributes to at most $(k+1)^2\ell^2$ such crossings, and the sum of the multiplicities of the intersection points in~$D$ is at most~$r$.  So the total number of such crossings is at most~$(k+1)^2r^2$.
  \qedhere
  \end{itemize}
  \end{proof}

  In the converse direction:
  \begin{lemma}\label{lem:rot2}
  If some colored drawing~$D''$ of~$G''$ belongs to~$\DD''$, then there is a drawing~$D$ of~$G$ that belongs to~$\DD$ and such that, for each connected component~$G_0$ of~$G$, the rotation system of~$G_0$ in~$D$ is induced by~$\pi$ or~$\pi^{-1}$.
  \end{lemma}
  \begin{proof}
  We consider the situation in~$D''$.  For each $v\in V(G')$, the grid~$G''_v$ is colored red and is thus embedded on~$\surf$.  By Lemma~\ref{lem:GRSgrids}, it has at least one flat horizontal strip, denoted $S''_v$.  Now, let $e=uv$ be an edge of~$G'$.  The graph $G''_u\cup Q''_e\cup G''_v$ is embedded on~$\surf$, because it is only made of red edges and either green or blue edges.  Moreover, it contains a ($2k\times k$)-grid, whose top half is a subgrid of~$G''_u$ and whose bottom half is a subgrid of~$G''_v$, both halves being connected by~$Q''_e$; each of the vertical strips of this ($2k\times k$)-grid contains a cell of~$S''_u$ and a cell of~$S''_v$.  Again by Lemma~\ref{lem:GRSgrids}, one of these vertical strips, $S''_e$, must be flat.

  Let $K''_e$ be the connected component of~$G''$ containing~$G''_u\cup Q''_e\cup G''_v$.  From the above, we obtain that each cell~$C$ of $S''_u\cup S''_e\cup S''_v$ bounds a disk whose interior is not intersected by the image of~$K''_e$ (since $K''_e-C$ is connected by construction).  If that disk contains part of~$D''$, then it contains one or several entire connected components of~$G''$, which can all be moved out of the disk.  So, without loss of generality, we can assume that each cell~$C$ of $S''_u\cup S''_e\cup S''_v$ bounds a disk~$d(C)$ whose interior is avoided by~$D''$.  In particular, this implies that all the cells of~$S''_u$ and~$S''_v$ are consistently oriented, in the following sense: for all such cells~$C$, the boundaries of the disks~$d(C)$ are visited by the directed cycle~$C$ with the same orientation (clockwise or counterclockwise).
  
  For each original edge~$e=uv$ of~$G''$, let $U''_e$ be the path that is the concatenation of the (unique, red) vertical path, in~$G''_u$, from a top vertex of~$S''_u$ to the source of~$e$, $e$ itself, and the (unique) vertical path from the target of~$e$ to a top vertex of~$S''_v$.  After contracting all the strips $S''_v$ to a point, the paths $U''_e$ together form a drawing~$D$ of the graph~$G$ whose rotation system, restricted to each connected component, is induced by $\pi$ or~$\pi^{-1}$.  To conclude the proof, there remains to prove that $D$ belongs to~$\DD$.

  Let $D_1$ denote the drawing obtained from $D''$ by removing all red, green, and blue edges; by definition of~$\DD''$, we obtain that $D_1$ belongs to~$\DD$.  Moreover, $D_1$ is the detachment of~$D$.  Finally, since $\DD$ is detachably stable, we obtain that $D$ also belongs to~$\DD$, as desired.
\end{proof}

We can now conclude the proof of \Cref{thm:mainrot}.
\begin{proof}[Proof of \Cref{thm:mainrot}]
  The construction of~$G''$ clearly takes $\OO(k^2rn)$ time.  Using the fact that $\DD$ is stable, it is routine to check that $\DD''$ is stable.  Moreover, given a representation of a colored drawing of size at most~$j$, one can decide whether it belongs to~$\DD''$ in time $\OO(\delta(j)+j)$.

  We now bound the intersecting size of~$\DD''$.  A drawing $D''$ in~$\DD''$ has two types of intersection points: (1) the intersection points involving only edges with colors in~$\{1,\ldots,c\}$, whose multiplicities sum up to at most~$r$, the intersecting size of~$\DD$; (2) the crossing points between a green edge and a blue edge, and there are at most $(k+1)^2r^2$ such crossings, whose multiplicities sup up to at most $2(k+1)^2r^2$.  So the intersecting size of~$\DD''$ is at most $r'':=r+2(k+1)^2r^2=\OO(k^2r^2)$.  

  Finally, we can apply \Cref{thm:main} to the class~$\DD''$ and the graph~$G''$, noting that $k=9s+3$.  Using the quadratic-time algorithm, the resulting runtime is $2^{\poly(c+s+r)}\cdot\big(\delta(\OO(cs^2r^2))+\OO(cs^2r^2)\big)\cdot n^2$, which is bounded from above by $2^{\poly(c+s+r)}\cdot\delta(\OO(cr^2s^2))\cdot n^2$.  Alternatively, the linear-time algorithm runs in $f(c+s+r)\cdot n$ time for some function~$f$.  This allows us to decide whether $G''$ has a colored drawing in~$\DD''$.  By Lemmas \ref{lem:rot1} and~\ref{lem:rot2}, this is equivalent to the fact that $(G,\pi)$ is a positive instance of the \DCROT{}.
\end{proof}

We remark that the \DCROT{} admits a formulation for non-orientable surfaces.  In this realm, an instance of this problem is given not only by a colored graph~$G$ and cyclic permutations~$\pi_v$ for each vertex~$v\in V(G)$, but also by a \emph{signature} $\lambda:E(G)\to\pm1$, which prescribes whether traversing an edge ``reverses the orientation''; see Mohar and Thomassen~\cite[Section~3.3]{mt-gs-01}.  Exactly the same arguments as above lead to a version of \Cref{thm:mainrot} for non-orientable surfaces; we omit the details.

\subsection{Applications of \Cref{thm:mainrot}}

We now consider applications of \Cref{thm:mainrot} to well-studied crossing number problems with fixed rotation system.  The traditional crossing number with a fixed rotation system in the plane is \NP-hard~\cite{DBLP:journals/algorithmica/PelsmajerSS11}, and its parameterized complexity has not been considered so far.

A crossing number variant that uses a similar restriction is the {\em homeomorphic joint crossing number} problem, which is similar to the joint crossing number problem, but with constraints on the rotation systems.  More precisely, in that problem, two input graphs $G_1$ and $G_2$ are cellularly embedded in an orientable surface~$\surf$ without boundary (every face is homeomorphic to a disk) and the solution must be composed of embeddings of $G_1$ and $G_2$ that are each homeomorphic to the input one (the rotation systems must be the same, possibly up to reversal).  Given an integer~$t\ge0$, the goal is to determine whether there is a way to achieve this with at most $t$ crossings between $G_1$ and~$G_2$.  The problem is also \NP-hard if the Euler genus of~$\surf$ is at least~12~\cite{DBLP:conf/isaac/HlinenyS15}.  The \emph{orientation-preserving joint crossing number} is similar, except that we restrict ourselves to orientable self-homeomorphisms (the rotation systems of both graphs cannot be flipped).

To capture these problems, we introduce a variant of the \cccn{} where the input also prescribes the rotation system, but as before the input rotation system may be flipped independently for each connected component of~$G$.  Formally:
\smallskip

\Prob{\cccn{}~with Flippable Rotation System}%
{A surface~$\surf$; a symmetric matrix~$M$ of size $c\times c$ with nonnegative integer values; a colored graph~$G$ with colors in~$\{1,\ldots,c\}$; a rotation system of~$G$}%
{Is there a normal colored drawing $D$ of~$G$ in $\surf$ such that:\begin{itemize}\parskip0pt\item for each pair $i,j\in\{1,\ldots,c\}$, the number of crossings in $D$ involving two edges, one colored $i$ and the other~$j$, is at most~$M_{i,j}$, and
\item for each connected component $G_0\subseteq G$, the rotation system of~$G_0$ in~$D$ is induced either by $\pi$ or by $\pi^{-1}$?
\end{itemize}}

We have:
\begin{proposition}\label{prop:CCCNrot}
The \cccn{}~\textsc{with Flippable Rotation System} can be solved in time $2^{\poly(s+t)}\!\cdot n^2$ or in time $f(s+t)\cdot n$ for some function~$f$, where $s$ is the topological size of~$\surf$, $t=\sum_{1\leq i\leq j\leq c}M_{i,j}$, and $n$ is the size of~$G$.
\end{proposition}
\begin{proof}
  The proof is identical to that of \Cref{prop:CCCN}, with \Cref{thm:mainrot} applied instead of \Cref{thm:main}.  (Note that the class~$\DD$ defined in the proof is also detachably stable.)
\end{proof}

As immediate corollaries, we obtain:
\begin{corollary}\label{cor:rotcr}
Let $t\ge0$ be an integer and let $G$ be a graph of size~$n$ with a rotation system.  Solving the traditional crossing number with fixed rotation system for~$G$, namely, deciding whether $G$ can be drawn in the plane with that rotation system and with at most $t$ crossings, can be done in time $2^{\poly(t)}\cdot n^2$ or in time $f(t)\cdot n$ for some function~$f$.
\end{corollary}
\begin{proof}
  We apply the algorithm of \Cref{prop:CCCNrot} for $G$ with the given rotation system, a single color, and the matrix $M=(t)$. If its outcome is negative, then the same answer clearly holds for our problem.
  Otherwise, we obtain a drawing of $G$ with at most $t$ crossings in the plane, in which the rotation system of each connected component is induced either by $\pi$ or $\pi^{-1}$.  Since, in an optimal drawing in the plane, two distinct connected components should not intersect, we can independently flip back the orientation of the components of $G$ which are flipped in the solution, and the outcome of our problem is thus positive.
\end{proof}

\begin{corollary}\label{cor:jointcrh}
  One can decide whether the homeomorphic joint crossing number of graphs $G_1$ and $G_2$, each cellularly embedded in a connected, orientable surface~$\surf$ without boundary, is at most~$t$ in time $2^{\poly(s+t)}\cdot n_0^2$ or in time $f(s+t)\cdot n_0$ for some function~$f$, where $s$ is the topological size of~$\surf$ and $n_0$ is the size of~$G_1\cup G_2$.
\end{corollary}
\begin{proof}
  $G_1$ and~$G_2$ are connected, since they are cellularly embedded on a connected surface.  The result is thus a direct application of~\Cref{prop:CCCNrot} in which $G=G_1\cup G_2$, for $i=1,2$, the edges coming from~$G_i$ are colored~$i$, and $M=\begin{pmatrix}0&t\\t&0\end{pmatrix}$.
\end{proof}

For the orientation-preserving version, we however get a slower runtime:
\begin{corollary}\label{cor:jointopcr}
  One can decide whether the orientation-preserving homeomorphic joint crossing number of graphs $G_1$ and $G_2$, each cellularly embedded in a connected, orientable surface~$\surf$ without boundary, is at most~$t$ in time $2^{\poly(s+t)}\cdot n_0^4$ or in time $f(s+t)\cdot n_0^3$ for some function~$f$, where $s$ is the topological size of~$\surf$ and $n_0$ is the size of~$G_1\cup G_2$.
\end{corollary}
\begin{proof}
  The disconnected graph~$G:=G_1\cup G_2$ comes equipped with a rotation system but, unlike in the proof of \Cref{cor:jointcrh}, we have to ensure that the components of $G$ cannot be flipped.
  We extend~$G$ to a connected graph~$\hat G$ with rotation system, by subdividing an edge $e_1$ of~$G_1$ and an edge $e_2$ of~$G_2$, thus inserting vertices $v_1$ and~$v_2$, and by then adding a new edge~$e$ between $v_1$ and~$v_2$.  There are $\OO(n_0^2)$ possibilities for the choice of~$\hat G$; we emphasize that it is not sufficient to choose the edges $e_1$ and~$e_2$, but we also need to select the rotation system at $v_1$ and~$v_2$ (since both have degree three, there are two possibilities for each).

  For each such choice of~$\hat G$, we assign color~$i$ to the edges coming from~$G_i$, $i=1,2$ (including the subdivided edges), and we assign color~3 to the new edge~$e$.  Then we apply \Cref{prop:CCCNrot} with $M=\begin{pmatrix}0&t&0\\t&0&0\\0&0&0\end{pmatrix}$; this takes the claimed running time.

  Since $\hat G$ is connected, if at least one choice of~$\hat G$ gives a positive outcome, then we have a solution of the orientation-preserving homeomorphic joint crossing number.  Conversely, in any drawing of $G$ on~$\surf$ such that $G_1$ and~$G_2$ are both embedded, there exist two points $p_1$ and~$p_2$, on two edges $e_1$ and~$e_2$ of $G_1$ and~$G_2$ respectively, that can be connected by a path~$q$ that avoids the images of $G_1$ and~$G_2$.  Subdividing $e_1$ and~$e_2$, and drawing $v_i$ on~$p_i$ ($i=1,2$) and $e$ on~$q$, yields a drawing of a suitable choice of~$\hat G$.  The application of \Cref{prop:CCCNrot} to the corresponding graph~$\hat G$ gives a positive outcome.
\end{proof}

\section{Computing drawings for positive instances}\label{sec:compute}

In this section, we prove the following result.
\begin{theorem}\label{thm:explicit}
  For every positive instance of the \DCREM{} problem, we can compute a representation of the corresponding drawing, of size the size $n$ of the input graph times a polynomial in $c+s+r+p+q$,
  in time $2^{\poly(c+s+r+p+q)}\cdot\delta(\OO(s+r)c)\cdot n^2$ (i.e., without overhead compared to the runtime of the quadratic-time algorithm in \Cref{thm:main}).
\end{theorem}

We remark that all FPT algorithms for crossing number variants in this paper are reductions to the \DCREM{} problem.  So, for any positive instance of such a variant, \Cref{thm:explicit} allows us to compute a corresponding drawing.

We need some preliminaries.  A graph embedding on a surface is \emph{cellular} if all its faces are homeomorphic to open disks.  Cellular graph embeddings can be manipulated through their \emph{combinatorial maps}.  Algorithmically, refined data structures such as \emph{graph-encoded maps}~\cite{l-gem-82} (see also~\cite[Section~2]{e-dgteg-03}) are used.

The idea of the proof of Theorem~\ref{thm:explicit} is the following.  The algorithm by Colin de Verdière and Magnard~\cite{cm-faeg2-26} decides whether a graph embeds on a 2-complex (\Cref{thm:embedding}), but also for positive instances it is possible to compute a representation of an embedding~\cite[Theorem~9.1]{cm-faeg2-26}, as restated below.  There remains to turn this embedding into~$\complex$ into a drawing.
\begin{theorem}[{Colin de Verdière and Magnard~\cite[Theorem~9.1]{cm-faeg2-26}}]%
\label{thm:compute}
  In Theorem~\ref{thm:embedding}, if $G$ has an embedding into~$\complex$, and $\complex$ has no edge incident to three triangles and no isolated vertex, an embedding of~$G$ can be computed in time $2^{\poly(C)}\cdot n^2$.   In detail, an embedding of a graph~$H$ is computed where:
  \begin{itemize}
  \item $H$ is obtained from~$G$ by augmenting it with at most $2C$ vertices and at most $3C+2u$ edges, and performing at most $C$ edge subdivisions, where $u$ is the number of connected components of~$G$;
  \item the images of the vertices of~$H$ cover the singular points of~$\complex$;
  \item the restriction of~$H$ to the surface part of~$\complex$ is cellular, specified by its combinatorial map, and by which point is mapped to which singular point of~$\complex$;
  \item the restriction of~$H$ to the isolated edges of~$\complex$ is specified by the sequence of vertices and edges of~$H$ appearing along each isolated edge of~$\complex$.
  \end{itemize}
\end{theorem}
Some comments on this theorem are in order.  First, the requirements on the 2-complex~$\complex$ are satisfied for all 2-complexes considered in this paper.  Second, it is necessary to add vertices to~$H$ in order to cover the singular points of~$\complex$, to add edges in order to make the graph cellularly embedded on the surface part of~$\complex$, and to subdivide edges in order to ensure that each edge of the graph is either entirely inside, or entirely outside, the surface part of~$\complex$.

\begin{proof}[Proof of \Cref{thm:explicit}]
  We run the quadratic-time algorithm of \Cref{thm:main} on the given instance $(G,p,q_1,\ldots,q_c)$ of the \DC{} problem.  Recall from the proof of that theorem that the corresponding graph~$G_2$ embeds into at least one 2-complex~$\complex_5\in\Gamma$, where $G_2$ has size $\OO((c+r+p+q)n)$ and $\complex_5$ has size $\OO(s+(c+r+p+q)r)$.  We then apply~\Cref{thm:compute}, obtaining, without overhead in the runtime, a graph~$H$ embedded on~$\complex_5$ whose size is that of~$G$ times a polynomial in $c+s+r+p+q$.

  The proof basically follows the same steps as the proof of \Cref{lem:embedding-to-drawing}.  First, up to modifying the embedding, we can assume that each necklace of~$\complex_5$ is either not used at all, or used by a necklace of~$H$ of the same type.  We thus have an embedding of a subgraph~$H_1$ of~$H$ (a supergraph of~$G_1$, in which some edges may be subdivided) in such a way that each isolated edge of~$\complex_4$ coincides with an edge of~$G_1$ of the same color, or is not used at all.

  In a second step, we perform the vertex splits as described in the proof of \Cref{lem:embedding-to-drawing}, resulting in an embedding in~$\complex_3$  of a graph~$H''$ obtained from~$H_1$ by at most $p$ vertex splits.  A minor detail: In the case where a singular point~$z_i$, $i\in S$, is used by the relative interior of an edge~$e=u$ of~$H$, the proof of \Cref{lem:embedding-to-drawing} performs a split of one of its vertices, say~$v$, as prescribed, by replacing $uv$ with $uv'$ where $v'$ is a new vertex and shrinking $uv'$ without moving~$u$ so that in the new embedding, $uv'$ does not contain any singular point of~$\complex_4$ in its interior.  In order to keep the graph cellular on the surface part, we cannot remove the edges on the surface part of the complex during this process, but now mark them as extra edges.

  In a third step, we remove the edges of the graph~$H''$ using edges in $\complex_3\setminus\complex_2$.  The resulting graph~$H'$, embedded into~$\complex_2$, is a supergraph, possibly subdivided, of a graph~$G'$ obtained from~$G$ by at most $p$ vertex splits and by removing, for each~$i$, at most $q_i$ edges of color~$i$.  Moreover, again, each isolated of~$\complex_2$ of color~$i$ coincides with an edge of~$G'$ of color~$i$, or is not used at all.  
  This embedding of~$H'$ in~$\complex_2$ turns into an embedding, in~$\surf$, of a graph~$H$ containing a drawing~$D'$ of~$G'$ in~$\DD$ by the following steps:
  \begin{enumerate}
  \item turn $\complex_2$ into~$\surf$ by reversing the process described in \Cref{lem:embedding-to-drawing}, namely, removing the isolated edges and gluing back the edges along which $\surf$ was cut; this results in a graph~$H$ cellularly embedded on~$\surf$;
  \item reinstate the edges of~$G'$ that were using isolated edges of~$\complex_2$ as walks in the set of edges obtained by gluing in the previous step.
  \end{enumerate}
  The graph~$H$ is specified by its combinatorial map.  If necessary, we push the image of~$G'$ away from the boundary of~$\surf$, by adding more vertices and edges (at most doubling the size of the graph).  Then we triangulate each face of~$H$.  We obtain a representation $(T,W)$ of the desired drawing of~$G'$ of size $n$ times a polynomial in $c+s+r+p+q$. 
\end{proof}

\section{Conclusions}\label{sec:concl}

\subparagraph{More applications.}
\Cref{thm:main} allows for almost endless combinations for crossing number variants, in terms of drawing styles, ways to count crossings, allowed vertex splits and edge removals, all possibly taking edge colors into account, all on an arbitrary surface.
\Cref{thm:mainrot} additionally makes it possible to restrict the rotation system in a rich subset of the variants.
This shows the versatility of our approach. Altogether, our general framework encompasses most existing crossing number variations, and even more general ones, implying, in a unified way, fixed-parameter tractable algorithms for these crossing number problems in any surface.

\subparagraph{Limitations.}  Nonetheless, our framework requires an upper bound~$r$ on the intersecting size of~$\DD$, and $\DD$ must be preserved by the removal of any vertex or edge.  The first requirement immediately excludes, e.g., the local crossing number (the minimum $t$ such that a given graph has a $t$-planar drawing, which is \NP-complete to compute already for~$t=1$~\cite{DBLP:journals/siamcomp/CabelloM13,DBLP:journals/algorithmica/GrigorievB07}), and the second requirement excludes some recently introduced drawing styles such as, for instance, $1^+$ and $2^+$-real face drawings~\cite{DBLP:journals/access/BinucciBDDHKLMT24}.  It is conceivable that Courcelle-based approaches would be still able to handle $1^+$ and $2^+$-real face drawings, albeit not easily.

Another weakness, in particular compared to the recent preprint by Hamm, Klute, and Parada~\cite{DBLP:journals/corr/abs-2412-13092}, is our quite restricted ability to handle predrawn parts of the input graph, essentially limited to fixing \emph{uncrossable} parts of the graph via rigid subembeddings, and to special restrictions like the fixed rotation system in \Cref{thm:mainrot}.  In contrast, \cite{DBLP:conf/compgeom/HammH22} and \cite{DBLP:journals/corr/abs-2412-13092} are very general in this respect and, in particular, allow crossings of fixed parts with the unfixed rest of the graph.

\subparagraph{Remarks and possible extensions.}   Beyond their use for strongly fan-planar drawings, surfaces with boundaries can be useful for other problems, because they allow to pinpoint specific regions of the surface.  One could actually consider a version with colored boundaries, reinforcing the definition of a stable class by requiring that it is preserved by color-preserving self-homeomorphisms.  Our arguments carry through; this only affects the running time by~a~$b!$ factor, where $b$ is the number of boundary components (see the proof of Lemma~\ref{lem:enum}).

Finally, a possible extension of the \DC{} problem would be to replace the surface~$\surf$ with an arbitrary 2-complex.  We leave open whether such an extension is possible, and also potential applications.

\bibliographystyle{plainurl}
\bibliography{bib-esa.bib}

\end{document}